\ifdefined\LONGVERSION{}
  \relax
\else
 \newcommand{\LONGVERSION}[1]{}
 \newcommand{\SHORTVERSION}[1]{#1}
\fi
\newcommand{\LONGSHORT}[2]{\LONGVERSION{#1}\SHORTVERSION{#2}}

\documentclass[acmsmall,screen,nonacm]{acmart}

\setcopyright{none}

\usepackage{booktabs}   
\usepackage{cancel}
\usepackage{caption}
\usepackage{changepage}
\usepackage{mathpartir}
\usepackage{cleveref}
\usepackage{cprotect}
\usepackage[shortlabels]{enumitem}
\usepackage{graphicx}
\usepackage{listings}
\usepackage{multicol}
\usepackage{multirow}
\usepackage{natbib}
\bibpunct{[}{]}{;}{a}{}{,}

\usepackage{needspace}
\usepackage{proof}
\usepackage{subcaption} 
\usepackage{tikz}
\usetikzlibrary{calc}
\usepackage[textsize=scriptsize,textwidth=1.2cm]{todonotes}
\usepackage{wrapfig}
\usepackage{xspace}

\newcommand{\elevator}{\textsc{Elevator}\xspace}
\newcommand{\moebius}{\textsc{Moebius}\xspace}

\lstdefinelanguage{Elevator}%
{%
  keywords=%
  {%
    box, %
    thunk, force,%
    store,load,let,in,%
    Type,%
    PolyNat,polyZero,polySuc,%
    Nat,zero,suc,%
    List,nil,cons,%
    Array,%
    fun,%
    if,then,else,%
    match,with,private%
  },%
  morecomment=[l]{--},%
  sensitive,
}
\LONGSHORT{ 
}{
  
}

\newcommand{\andothers}{\mathop{\ldots}\xspace}

\newcommand{\syntaxDef}{\mathbin{:=}\xspace}
\newcommand{\syntaxAlt}{\mathbin{\ \mid\ }\xspace}

\newcommand{\ms}[1]{\mathsf{#1}}
\newcommand{\mmode}[1]{{\mathchoice{\ms{#1}}{\ms{#1}}{\scriptscriptstyle\ms{#1}}{\scriptscriptstyle\ms{#1}}}}

\newcommand{\mGF}{\mmode{GF}}

\newcommand{\mP}{\mmode{P}}
\newcommand{\mC}{\mmode{C}}

\newcommand{\modeSet}{\mathcal{M}}

\newcommand{\modeSigSymb}{\mathcal{S}}
\newcommand{\modeSig}[1]{\modeSigSymb({#1})}
\newcommand{\modeCo}{\mathtt{C}}
\newcommand{\modeWk}{\mathtt{W}}

\newcommand{\ctxEmpty}{{\ \cdot\ }}
\newcommand{\ctxCons}[2]{#1 \mathbin{,} #2}

\newcommand{\subEmpty}{{\ \cdot\ }}
\newcommand{\subCons}[2]{#1 \mathbin{,} #2}

\newcommand{\kiUpSymb}{{\uparrow}}
\newcommand{\kiTypeSymb}{{\mathtt{Type}}}

\newcommand{\kiCtxUp}[4]{{[#3\ \vdash\ #4]\kiUpSymb}^{#1}_{#2}}
\newcommand{\kiType}[1]{\kiTypeSymb_{#1}}

\newcommand{\tyBoxSymb}{{\Box}}

\newcommand{\tyUpSymb}{{\uparrow}}
\newcommand{\tyDownSymb}{{\downarrow}}
\newcommand{\tyTensorSymb}{\mathbin{\otimes}}
\newcommand{\tyPlusSymb}{\mathop{\oplus}}
\newcommand{\tyAndSymb}{\mathop{\&}}
\newcommand{\tyOneSymb}{{\mathtt{\mathbf{1}}}}
\newcommand{\tyLarrSymb}{\mathbin{\multimap}}
\newcommand{\tyArrSymb}{\mathbin{\rightarrow}}

\newcommand{\tyUp}[3]{#3 \tyUpSymb^{#1}_{#2}}
\newcommand{\tyCtxUp}[4]{{[#3\ \vdash{} \ #4] \tyUpSymb}^{#1}_{#2}}
\newcommand{\tyDown}[3]{#3 \tyDownSymb^{#1}_{#2}}
\newcommand{\tyOne}[1]{\tyOneSymb_{#1}}
\newcommand{\tyLarr}[2]{#1 \tyLarrSymb{} #2}
\newcommand{\tyArr}[2]{#1 \tyArrSymb{} #2}
\newcommand{\tyForall}[3]{(#1 {:} #2) \tyArrSymb{} #3}
\newcommand{\tyThunk}[1]{\ulcorner{} #1 \urcorner}
\newcommand{\tyCtxThunk}[2]{\tyThunk{#1\ .\ #2}}

\newcommand{\tyForce}[1]{\llcorner{} #1 \lrcorner}
\newcommand{\tyForceWith}[2]{\tyForce{#1}@#2}

\newcommand{\tmInName}{{\mathsf{in}}}

\newcommand{\tmSusp}[1]{\ulcorner{} #1 \urcorner}
\newcommand{\tmCtxSusp}[2]{\tmSusp{#1\ .\ #2}}
\newcommand{\tmCtxSuspElab}[4]{\tmSusp{#3\ .\ #4}^{#1}_{#2}}

\newcommand{\tmForce}[1]{\llcorner{} #1 \lrcorner}
\newcommand{\tmForceWith}[2]{\tmForce{#1}@#2}
\newcommand{\tmForceWithElab}[4]{\tmForce{#3}^{#1}_{#2}@#4}
\newcommand{\tmStoreName}{{\mathsf{store}}}
\newcommand{\tmStore}[1]{\tmStoreName\ (#1)}
\newcommand{\tmStoreElab}[3]{\tmStoreName^{#1}_{#2}\ (#3)}
\newcommand{\tmLoadName}{{\mathsf{load}}}
\newcommand{\tmLoad}[3]{\tmLoadName\ (#1) = #2\ \tmInName\ #3}
\newcommand{\tmLoadElab}[5]{\tmLoadName^{#1}_{#2}\ (#3) = #4\ \tmInName\ #5}
\newcommand{\tmOne}{{\mathsf{one}}}
\newcommand{\tmTLam}[3]{\Lambda{} (#1 {:} #2) . #3}
\newcommand{\tmLam}[3]{\lambda{} (#1 {:} #2) . #3}
\newcommand{\tmTApp}[2]{#1\mathbin{}#2}
\newcommand{\tmApp}[2]{#1\mathbin{}#2}

\newcommand{\ctxMerge}[2]{#1 \bowtie{} #2}
\newcommand{\eraseForSub}[1]{{(#1)}^{-}}
\newcommand{\sub}[2]{[#1]#2}
\newcommand{\subNorm}[3]{{[#1]}^{#2}#3}
\newcommand{\subNeut}[3]{{[#1]}^{#2}_{\mathtt{R}}#3}
\newcommand{\subMerge}[2]{#1 \bowtie{} #2}

\newcommand{\ctxCls}{\ \mathtt{ctx}}
\newcommand{\kindCls}{\ \mathtt{kind}}

\newcommand{\stepto}{\mathbin{\longrightarrow}}
\newcommand{\refinestepto}[1]{\mathbin{\rightsquigarrow^{#1}}}

\newcommand{\loadStoreRed}{\mathsf{load}{-}\mathsf{store}}
\newcommand{\forceSuspRed}{\mathsf{force}{-}\mathsf{susp}}

\begin{document}

\title{Polymorphic Metaprogramming with Memory Management}
\subtitle{An Adjoint Analysis of Metaprogramming}

\author{Junyoung Jang}

\affiliation{
  \department{School of Computer Science} 
  \institution{McGill University}         
  \country{Canada}                    
}
\email{junyoung.jang@mail.mcgill.ca}          

\author{Brigitte Pientka}

\affiliation{
	\department{School of Computer Science} 
	\institution{McGill University}         
	\country{Canada}                    
}
\email{bpientka@cs.mcgill.ca}         


\begin{abstract}
We describe \elevator, a unifying polymorphic foundation for metaprogramming
with memory management based on adjoint modalities.
In this setting, we distinguish between multiple memory regions using
\emph{modes} where each mode has a specific set of structural
properties. This allows us not only to capture linear (i.e.\ garbage-free) memory
regions and (ordinary) intuitionistic (i.e.\ garbage-collected or
persistent) memory regions, but also to capture accessibility between the
memory regions using a preorder between modes. This preorder
gives us the power to describe monadic and comonadic programming. As a
consequence, it extends the existing logical view of metaprogramming
in two directions: first, it ensures that code generation can be
done efficiently by controlling memory accesses; second, it allows us to provide
resource guarantees about the generated code (i.e.\ code that is for
example garbage-free).


We present the static and dynamic semantics of \elevator.
In particular, we prove the substructurality of
variable references and type safety of the language.
We also establish mode safety, which guarantees that the evaluation of a term
does not access a value in an inaccessible memory.
\end{abstract}



\maketitle{}

\section{Introduction}\label{sec:intro}
Metaprogramming is a commonly used technique for generating or
manipulating a program within another program. This allows programmers
to automate error-prone or repetitive tasks and exploit
domain-specific knowledge to optimize the generated code.
Unfortunately, writing safe metaprograms remains very challenging and
is sometimes frustrating as, traditionally, errors in the generated code
are only detected when running it but not at the time when the code is
generated. 

The state-of-the-art in most existing languages \--- from typed
programming languages such as Haskell~\cite{Sheard:Haskell02} to
dependently typed languages such as Coq~\cite{Anand:ITP18} and
Agda~\cite{vanderWalt:IFL12} \--- is that the generated code is simply
syntax and lacks many guarantees that we would expect: it may contain
dangling free variables and it may be ill-typed. Even in languages
such as MetaML~\cite{Taha:TCS00,Taha:POPL03} or more recently
Scala~\cite{Parreaux:SPLASH17,Parreaux:POPL18} and
Typed Template Haskell~\cite{Xie:POPL22}, which support the
generation of typed code, we lack guarantees about resource and memory
usage which often prevents compilers from exploiting this information
to generate efficient code. The generation of polymorphic code also
remains challenging despite recent work by \citet{Xie:POPL22} and \citet{Jang:POPL22}.
How can we improve this state-of-the-art and provide a uniform logical
foundation for type-safe, polymorphic, resource-aware metaprogramming?

Recently, the question of how to combine multiple logics with different
structural properties including linear, affine, strict, and
  (ordinary) intuitionistic logics has been tackled through graded or
quantitative type systems (see 
\citep{McBride16,Atkey:LICS18,Moon:ESOP21,Choudhury:POPL21,Wood22esop,Abel:ICFP23}). The
essential idea is to track and reason explicitly about the usage of 
assumptions through grades. However, they do not support
metaprogramming. In particular, there is no distinction between code
(which simply describes syntax) and programs.
%

We pursue here an alternative based
on adjoint logic~\cite{Benton:CSL94,Reed09,Pruiksma:JLAMP21,Pruiksma24phd,Jang:FSCD24}.
In this adjoint setting, we distinguish between multiple logics using
\emph{mode}s where each mode has a specific set of structural
properties. For example, to model linear logic, we force every
resource in a mode to be used exactly once (that is, the mode
admits neither weakening nor contraction), thus ensuring
freedom from garbage. To model strictness, we specify a
mode that does not admit weakening. In addition, we allow a preorder
between modes $m \geq k$ that must be compatible with the structural
properties of $m$ and $k$, in the sense that \(m\) should at least allow
the structural rules allowed in \(k\). This allows us, among other aspects, to model
intuitionistic S4~\citep{Pfenning01mscs} and lax logic~\citep{Benton:JFP98},
representing both comonadic and monadic programming. 
To model the distinction between code that is being generated and the
program generator, we can use two modes: mode $\mC$ 
describes a suspended program that is eventually executable
at mode $\mP$ for programs where $\mC \geq \mP$.
The combination of adjoint modalities and a general mode structure
allows us to write linear comonadic programs and hence guarantees that the
generated code is garbage-free.

In this paper, we prsent a new foundation called \elevator{} for polymorphic
metaprogramming with memory management based on adjoint logic. It extends
prior work by \citet{Jang:FSCD24} in the following three directions:

\paragraph{1. Contextualizing Adjoint Logic}
In \elevator, modes refer to different regions in memory,
and the accessibility between regions is described using a preorder
of modes. For instance, \(m \geq k\) means memory region \(k\) can refer to memory
region \(m\). 
As a consequence, the accessibility relation is easily checkable
while being rich enough to express a wide range of
memory regions and their relationships. Following prior work by \citet{Jang:FSCD24},
we move between regions using adjoint modalities $\tyDownSymb$ (down-shift modality)
and $\tyUpSymb$ (up-shift modality).
First, we interpret the down-shift modality $\tyDown{m}{k}{A_m}$
(where $m \geq k$) as a \emph{pointer in region $k$ to an object of type $A_m$ in
  region $m$}. More importantly, we contextualize the up-shift modality 
$\tyCtxUp{m}{k}{\Gamma}{A_k}$ (where $m \geq k$).
This allows us to describe an
\emph{(open) suspended object (thunk) with free variables $\Gamma$ in
  the region $k$} and to shield it from evaluation by tagging it as
belonging to region $m$.
In the metaprogramming setting, $k$
corresponds to $\mP$ (programs) and $m$ corresponds to $\mC$ (code),
i.e. $\mC \geq \mP$. Contextual up-shift modalities hence allow us to
describe a suspended piece of program, which we call a
\emph{template}. The down-shift modality then creates a pointer in the
$\mP$ region that refers to a template in $\mC$.
This fine-grained distinction between the template and the pointer to it
is only implicit when we use a (contextual) necessity modality
from modal logic S4 to characterize code \cite{Davies:ACM01,Nanevski:TOCL08,Jang:POPL22}.
It is this fine-grained distinction that gives \elevator{} the ability and power to
generate code more efficiently by controlling which structures are
accessed when in memory.
Moreover, as modes also allow us to describe a range of
structural properties, 
 we can express even stronger guarantees: first, we can
 ensure that the suspended program is garbage-free; second, it also
 allows us to ensure that the program generation itself is garbage-free.
Hence, \elevator{} provides a uniform metaprogramming language to
express and statically verify a range of type and memory guarantees. 

\paragraph{2. An Operational Semantics for Adjoint Modalities Suitable for Metaprogramming} 
We define a suitable small-step operational semantics for \elevator{} which is
compatible with metaprogramming. The key idea for this operational semantics is
highlighted in the following principle:
\begin{quotation}
  The evaluation of a term should not know anything about
  a term in an inaccessible memory other than its syntactic structure.
\end{quotation}
This gives a different semantics from the one given in~\citet{Jang:FSCD24}
which treats a template as a black box. Since a template is an
object with a syntactic structure, we can keep the syntactic equality of
templates regardless of the construction order. For instance, whether we splice
\lstinline!x + x! into the hole \lstinline!X! of \lstinline!y * X! or splice
\lstinline!y! into the hole \lstinline!Y! of \lstinline!Y * (x + x)!, our system
returns exactly the same program template \lstinline!y * (x + x)!.
In fact, this property is also
satisfied in existing metaprogramming systems based on modal logic
S4~\cite{Davies:ACM01,Nanevski:TOCL08,Jang:POPL22}. 
This is achieved by treating splicing eagerly in our operational
semantics even in inaccessible regions. Hence, our system
conservatively extends prior metaprogramming foundations based on S4 with fine-grained 
memory control and access guarantees in contrast to~\citet{Jang:FSCD24}.

\paragraph{3. Adjoint Modalities and System-F style Polymorphism}
To obtain a core language for practical metaprogramming, we add
System-F style polymorphism to \elevator. This provides on the one
hand a logical foundation and alternative for combining linearity with
polymorphism as in, for example~\citet{Bernardy:POPL18}, but on the
other hand, supports polymorphic code generation as in~\citet{Jang:POPL22}.
To include System-F style polymorphism, we introduce the up-shift modality
for kinds and the concept of \emph{type templates}, which describe a
syntactic form of a type that can be used across modes.
This concept of type templates has a significant impact on
how we define the syntax of types and operations such as
substitution. First of all, splicing a type template into another
introduces a type-level redex. 
To avoid this kind of type-level computation and keeping syntactic
equalities between identical types, we define the syntax of types only for
normal types and define substitution on types as hereditary
substitution~\cite{Watkins02tr,Nanevski:TOCL08}, which reduces all
type-level redices occuring during the substitution. An orthogonal problem
arising due to type variables is how structural rules apply to type
variables. We take a similar approach 
to~\citet{Atkey:LICS18} for type variables, i.e.\ they persist regardless
of their modes. This approach is sensible because the type-checking
and type/kind construction are compile-time steps rather than runtime
computation. This approach is also convenient in terms of programming,
as the type of \lstinline!id : (\alpha : Type) -> \alpha -o \alpha!
where the type variable $\alpha$ occurs twice is well-formed in a
linear mode without any overhead.

On the theoretical side, we prove type safety for
\elevator{}. Moreover, 
we prove the substructurality of variable references
and prove mode safety, which guarantees that the evaluation of a term
does not access a value in an inaccessible memory.
On the practical side, we discuss
as an example the \emph{In-Place Array Update}. Here we update each
element of an array using a function $f$ that transforms an element of
type $\alpha$ to an element of type $\beta$. We demonstrate how to
utilize our mode structure to update the array in-place and how to
generate polymorphic, garbage-free code that is parametric in $f$
and polymorphic in $\alpha$ and $\beta$.
We have also implemented and checked this example with our prototype
 environment for \elevator.
In conclusion, our work lays the ground for a practical foundation for
metaprogramming with memory regions that extends System F with
adjoint modalities. 

\section{Motivations}\label{sec:motivations}
We give here a gentle introduction to metaprogramming with adjoint
modalities where we distinguish between 
code (mode $\mC$) and programs (mode $\mP$) where $\mC \geq \mP$. We
subsequently use modes to generate garbage-free code and capture cross-stage persistence of
polymorphic type variables.

\subsection{Code Generation Using Contextual Adjoint Modalities}



To illustrate how to write code generators using adjoint modalities,
we generate code to look up the \lstinline!n!-th element in a polymorphic list
\lstinline!xs! where \lstinline!xs! is supplied at a later (next)
stage. Given an
\lstinline!n!, we want the program \lstinline!nth! to return a pointer to the
suspended program template 
\begin{center}
\lstinline!t[\alpha, xs. head (!$\underset{n}{\underbrace{\mbox{\lstinline!tail (.. (tail !}}}$\lstinline! xs)))]t!
\end{center}

We first give a naive implementation which literally
translates the program from ~\citet{Jang:POPL22} to the adjoint
setting. 

\begin{lstlisting}
nth : Nat$_\mP$ -> t[\alpha:Type$_\mP$, xs:\alpha List$_\mP$ |- \alpha]t^|$^\mC_\mP$v|$^\mC_\mP$
nth n = match n with
  | 0 => store t[\alpha, xs . head xs]t
  | n =>
    load C = nth (n - 1) in
    store t[\alpha, xs . f[C]f@(\alpha, tail xs)]t
\end{lstlisting}

The function \lstinline!nth! generates a pointer to a suspended
program (i.e.\ a template) in region \(\mC\) that is still waiting
for the type variable $\alpha$ and the list 
\lstinline!\alpha List$_\mP$!. When we eventually supply the
concrete list for \lstinline!xs!, we can resume the evaluation of this
suspended program and retrieve the \lstinline!n!-th item in the given list \lstinline!xs!. 
What is the idea behind this program?
In the base case where we want to retrieve the zeroth element, we generate the template
\lstinline!t[\alpha, xs . head xs]t! where \lstinline!\alpha! and
\lstinline!xs! are variables provided when we actually evaluate the
term. Then, using the \lstinline!store! primitive, we return a pointer
to this template.
In the recursive case where \lstinline!n! is greater than \lstinline!0!, we want to retrieve the element at position \lstinline!n!. 
The recursive call \lstinline!nth (n - 1)! returns a pointer to a template that retrieves the element at position \lstinline!(n - 1)!.
We hence first load the template (i.e.\ suspended program) from the pointer and bind it to \lstinline!C!. Note that \lstinline!C! stands for an open program template that depends on $\alpha$ and \lstinline!xs!. To build the final program template that retrieves the element at position \lstinline!n!, we use \lstinline!C! instantiating \lstinline!xs! with \lstinline!tail xs!. We write \lstinline! f[C]f@(\alpha, tail xs)! for this instantiation of the open program template with the substitution \lstinline!@(\alpha, tail xs)!.
Last, we return a pointer to this program template as a result.

The program above, modulo some syntax difference, is identical to
the implementation in other metaprogramming systems based on the contextual modal
type, e.g.~\cite{Jang:POPL22}. However, note how we repeatedly store/load
program templates, as the recursive program returns in essence a pointer to
a program template. This incurs an overhead of repeated memory accesses during
code generation.

In \elevator{}, we can do better by taking advantage of the adjoint modalities
to minimize memory accesses during code generation. The key idea is to separate the
generation and composition of templates from the storing of the final result in region $\mP$.
To accomplish this we first define the function \lstinline!nthGen!. 
\begin{lstlisting}
nthGen : Nat$_\mC$ -> [\alpha:Type$_\mP$, xs:\alpha List$_\mP$ |- \alpha]^|$^\mC_\mP$
nthGen n = match n with
  | 0 => t[\alpha, xs . head xs]t
  | n => t[\alpha, xs . f[nthGen (n - 1)]f@(\alpha, tail xs)]t
\end{lstlisting}

Note that this function does not access memory region $\mP$.
In fact, this function entirely lives in mode $\mC$. 
However, moving the program generation safely to region $\mC$
also means that we need to move the input to $\mC$; it is now
\lstinline!Nat$_\mC$! (a natural number stored in memory
region $\mC$). We can accomplish this by explicitly converting
a natural number from region $\mP$ to one in region $\mC$.

\begin{lstlisting}
convertNat : Nat$_\mP$ -> Nat$_\mC$ v|$^\mC_\mP$
convertNat n = match n with
  | 0 => store 0
  | n => load N = convertNat (n - 1) in
         store (N + 1)
\end{lstlisting}

\newpage
Finally, we re-implement the code generator \lstinline!nth!.
\begin{lstlisting}
nth : Nat$_\mP$ -> [\alpha:Type$_\mP$, xs:\alpha List$_\mP$ |- \alpha]^|$^\mC_\mP$v|$^\mC_\mP$
nth n =
  load N = convertNat n in
  store (nthGen N)
\end{lstlisting}

Here, we first convert the natural number \lstinline!n! in $\mP$
into $\mC$ and then generate the result in mode $\mC$. Once we get the
final program template, we return a pointer to it. This clearly separates
template generation and composition from storing and loading templates.

At first glance, it may seem that we just move repeated memory accesses
using \lstinline!store!/\lstinline!load! instructions into
\lstinline!convertNat! and that this only increases the number of iterations
from \(n\) to \(2n\). However, the situation is now very different from the
previous, unoptimized version:

\begin{itemize}
\item First, now we do not need to store and load program templates which
  are more complex structures than storing and loading natural numbers.
  Storing the latter requires less memory space than storing the program
  templates. This reduces IO cost, which dominates the computation cost in
  most modern architectures.
\item Furthermore, we can use efficient data representation for natural numbers.
  For instance, in the case of \lstinline!Nat$_\mP$!, we can implement the
  function based on binary encoding, i.e. instead of using a recursive call
  on \lstinline!n - 1!, we can use one that applied to \lstinline!n/2!.
\item Finally, if we use the same natural number to generate many
  different templates, the loop for \lstinline!convertNat! can be shared
  after inlining. Then, the natural number is only converted once.
\end{itemize}

In other words, adjoint modalities open up more opportunities for optimizing
code generation.

\subsection{Polymorphic Code Generation and Cross-Stage Persistent Type Variables}
\label{subsec:motivation-poly-code-gen}
To illustrate the issues of supporting polymorphic code generation and
enforcing resource guarantees, recall the polymorphic \lstinline!map! function:

\begin{lstlisting}
map: (\alpha,\beta : Type) -> (\alpha -> \beta) -> \alpha list -> \beta list
map f xs = match xs with
  | Nil       => Nil
  | Cons x xs => Cons (f x) (map f xs)
\end{lstlisting}

We will transform this program step-wise.
As a first step, we ask: how would the type and the program change if we assume that the input
list \lstinline!\alpha list! is linear (i.e.\ garbage-free) and we want to preserve this
linearity? We consider here linear lists as a simplified version of an
array to illustrate the considerations that come into play. We give a
more realistic version for arrays in \Cref{sec:implementation}.
In the second step, we generate polymorphic code; based on the input list
\lstinline!\alpha list!, we generate a polymorphic template that is still waiting
for \lstinline!f : \alpha -> \beta! and builds the transformed
\lstinline!\beta list!.

\paragraph{Enforcing linearity}
Enforcing linearity hits two roadblocks. First,
the program \lstinline!map! is not linear in the function
\lstinline!f : \alpha -> \beta!. In the base case where we have
an empty list, \lstinline!f! is unused whereas in the step case, we use \lstinline!f!
twice. Second, the type of \lstinline!map! uses type
variables $\alpha$ and $\beta$ in the garbage-free parts
(\lstinline!\alpha list! and \lstinline!\beta list!)
as well as in the unrestricted argument (\lstinline!\alpha -> \beta!).

To distinguish between a garbage-free memory region ($\mGF$) (i.e.\ every
resource is used exactly once) and those regions $\mP$ where resources can
be used freely, we use the preorder $\mP \geq \mGF$. We now decorate
type variables as well as types with the appropriate adjoint modalities.
For \lstinline!f!, we put \lstinline!^|$^\mP_\mGF$v|$^\mP_\mGF$! on its type,
so that it becomes closed in terms of garbage-free memory region and thus
can be used multiple times (as many times as the input list requires).
Here, since we want to make \lstinline!f! closed, we use
\lstinline!^|$^\mP_\mGF$! without a context.
For type variables, we put \lstinline!^|$^\mP_\mGF$! on their kind.
With this adjustment, we make the type arguments \lstinline!\alpha! and
\lstinline!\beta! be templates for types, which are available both in
the garbage-free memory region and unrestricted region.
In order to use these templates as real types, we add \lstinline!f[]f!
around the places where we use the type variables.

\begin{lstlisting}
mapLin : (\alpha, \beta:Type$_\mGF$^|$^\mP_\mGF$) -> (f[\alpha]f -o f[\beta]f) ^|$^\mP_\mGF$v|$^\mP_\mGF$ -o f[\alpha]f List$_\mGF$  -o f[\beta]f List$_\mGF$
mapLin f xs = match xs with
  | Nil$_\mGF$       => load F = f in Nil$_\mGF$     -- "load" discards F
  | Cons$_\mGF$ x xs =>
     load F = f in                        -- "load" allows to use F twice
     Cons$_\mGF$ (f[F]f x) (mapLin (store(F)) xs)
\end{lstlisting}

In the program, we load a template
\lstinline!F : (f[\alpha]f -o f[\beta]f) ^|$^\mGF_\mGF$!
from the pointer \lstinline!f!. This consumes \lstinline!f!
in \(\mGF\) while introducing \lstinline!F! lives in \(\mP\).
As \(\mP\) allows multiple uses of objects in its memory region,
we can both skip \lstinline!F! completely (in the base case) or use
it twice (in the step case). In the second use of the step case,
we use \lstinline!store(F)! to get a pointer in \(\mGF\) referring to
\lstinline!F!, because the old pointer \lstinline!f! is already used
and cannot be used twice.

\paragraph{Generating code  with resource guarantees}
As \elevator{} allows a program to have more modes than \(\mGF\)
and \(\mP\), we can convert \lstinline!mapLin! into a program
generator which we will call \lstinline!mapLinMeta!. There are a few more
ingredients that we need. First, we need to add a mode \(\mC\) for code
with the preorder \(\mC \ge \mP\). We also need to add another argument
to the function that lifts a value of \lstinline!f[\alpha]f! to the
$\mC$ level. This is necessary since we need to splice each entry
into the code of the transformed list.
Parameterizing polymorphic code generators with a polymorphic
lift function has also been used in other metaprogramming
systems such as \moebius \cite{Jang:POPL22}. 

As we need to use this lifting function multiple times at the program
level, we put it under the modalities \lstinline!^|$^\mP_\mGF$v|$^\mP_\mGF$!. 
We also need to adjust the up-shift for one of the type arguments: for \(\alpha\),
we want to use it in the current stage (for the type of the lifting function
and input list type) and in the next stage (for the type of the mapping function
\lstinline!f!). To make it available not only for garbage-free and unrestricted
modes but also in a cross-stage persistent manner, we use \lstinline!^|$^\mC_\mGF$!
instead of \lstinline!^|$^\mP_\mGF$!. Note that $\mC > \mP$ and hence
the type template is globally available (i.e.\ has fewer assumptions about its use site).
Also, note that we do not need to provide \(\beta\) upfront. In fact, if we
want to generate the most polymorphic program template, we should allow \(\beta\)
to be chosen later when the generated program is used. Putting these different
considerations together, we get the following \lstinline!mapLinMeta!.

\begin{lstlisting}
mapLinMeta : (\alpha:Type$_\mGF$^|$^\mC_\mGF$)
           -> (f[\alpha]f -o f[\alpha]f^|$^\mC_\mGF$v|$^\mC_\mGF$)^|$^\mP_\mGF$v|$^\mP_\mGF$                      % Generic lift function
           -o f[\alpha]f List$_\mGF$                                   % Input list xs
           -o [\beta:Type$_\mGF$^|$^\mP_\mGF$, f:(f[\alpha]f -o f[\beta]f)^|$^\mP_\mGF$v|$^\mP_\mGF$ |- f[\beta]f List$_\mGF$]^|$^\mC_\mGF$v|$^\mC_\mGF$
mapLinMeta \alpha lift xs = match xs with
  | Nil$_\mGF$       => load LIFT = lift in store t[\beta, f . load F = f in Nil$_\mGF$]t
  | Cons$_\mGF$ x xs => load LIFT = lift in
                  load X = LIFT x in
                  load YS = mapLinMeta \alpha (store (LIFT)) xs in
                  store t[\beta, f .
                             load F = f in Cons$_\mGF$ (f[F]f f[X]f) f[YS]f@(\beta, store(F))]t
\end{lstlisting}
In both cases in the above program, 
we first access the template that \lstinline!lift! points to using
\lstinline!load!. Hence, we can free the memory location for the pointer
\lstinline!lift! itself, and we are not restricted in the usage of
\lstinline!LIFT!. The base case returns the base case of
\lstinline!linearMap! as a program template. In the step case, we first use
\lstinline!LIFT! on \lstinline!x! to move it to the $\mC$ region as a template
representing its value (\lstinline!X!), and then construct the head
(\lstinline!f[F]f f[X]f!). The tail of the resulting list is generated by the
recursive call (which gives \lstinline!YS! after loading the result) and then
instantiating \lstinline!YS! using \lstinline!@(\beta, store(F))!.



We can also optimize this code generator by exploiting the same ideas as
in the earlier \lstinline!nth! example to get a more efficient version.
\begin{lstlisting}
mapLinMetaGen : (\alpha:Type$_\mGF$^|$^\mC_\mGF$)
                 -> (f[\alpha]f ^|$^\mC_\mGF$) List$_\mC$
                 -> [\beta:Type$_\mGF$^|$^\mP_\mGF$, f:(f[\alpha]f -o f[\beta]f)^|$^\mP_\mGF$v|$^\mP_\mGF$ |- f[\beta]f List$_\mGF$]^|$^\mC_\mGF$
mapLinMetaGen \alpha Xs = match Xs with
  | Nil$_\mC$       => store t[\beta, f . load F = f in Nil$_\mGF$]t
  | Cons$_\mC$ X Xs => store t[\beta, f .
                     load F = f in Cons$_\mGF$ (f[F]f f[X]f) f[mapLinMetaGen \alpha Xs]f@(\beta, store(F))]t
\end{lstlisting}
Here, we use \lstinline!->! for \lstinline!mapLinMetaGen!, as it is a function
in \(\mC\), not in the garbage-free mode \(\mGF\). We follow the same routine
as in the \lstinline!nth! optimization. We add a helper function
\lstinline!mapLinMetaGen! that only cares about how to compose
a program template.

We also need a function that converts a list. This can be defined as the following:
\begin{lstlisting}
convertList : (\alpha:Type$_\mGF$^|$^\mC_\mGF$)
            -> (f[\alpha]f -o f[\alpha]f ^|$^\mC_\mGF$v|$^\mC_\mGF$)^|$^\mP_\mGF$v|$^\mP_\mGF$
            -o f[\alpha]f List$_\mGF$
            -o (f[\alpha]f ^|$^\mC_\mGF$) List$_\mC$ v|$^\mC_\mGF$
convertList \alpha lift xs = match xs with
  | Nil$_\mGF$       => load LIFT = lift in store Nil$_\mC$
  | Cons$_\mGF$ x xs => load LIFT = lift in
                   load X = LIFT x in
                   load XS = convertList \alpha (store(LIFT)) xs in
                   store (Cons$_\mC$ X XS)
\end{lstlisting}

Now, we can define the main function by first converting the input list
(using \lstinline!convertList!) and then building a template using
\lstinline!mapLinMetaGen! on that converted list.
\begin{lstlisting}
mapLinMeta : (\alpha:Type$_\mGF$^|$^\mC_\mGF$)
             -> (f[\alpha]f -o f[\alpha]f ^|$^\mC_\mGF$v|$^\mC_\mGF$)^|$^\mP_\mGF$v|$^\mP_\mGF$
             -o f[\alpha]f List$_\mGF$
             -o [\beta:Type$_\mGF$^|$^\mP_\mGF$, f:(f[\alpha]f -o f[\beta]f)^|$^\mP_\mGF$v|$^\mP_\mGF$ |- f[\beta]f List$_\mGF$]^|$^\mC_\mGF$v|$^\mC_\mGF$
mapLinMeta \alpha lift xs = 
  load XS = convertList \alpha lift xs in
  store (mapLinMetaGen \alpha XS)
\end{lstlisting}

As in the \lstinline!nth! case, the conversion function \lstinline!convertList!
owns most of \lstinline!store!/\lstinline!load! pairs. Thus, if we reuse the
result of the \lstinline!convertList! call, we can avoid a lot of
\lstinline!store!/\lstinline!load! operations. For example, if we generate
a program using \lstinline!mapLinMeta! on a list and also generate a program
to fold values in the same list, we can inline those generators and share
the result of \lstinline!convertList!.

In summary, this example shows that we can seamlessly combine metaprogramming (with
a cross-stage persistent type variable) and garbage-free memory management. Moreover,
we can still apply the same optimization routine we apply for metaprogramming for such
a combination.

\section{\elevator: Polymorphic Metaprogramming Using Adjoint Modalities}
\subsection{Mode Specification}
\elevator{} is a system parametrized by a mode specification.
A mode specification \(\modeSet\) is a preordered set of modes with a \emph{signature} operation
\(\modeSigSymb : \modeSet \to \mathcal{P}(\{\modeCo, \modeWk\})\) where
\(\mathcal{P}(S)\) means the powerset of \(S\). This signature operation tells us what
structural rules (\(\modeCo\) for contraction and \(\modeWk\) for weakening) are allowed
in each mode. The signature operation should satisfy the following condition:
\[
  m \geq k \text{ then } \modeSig{m} \supseteq \modeSig{k}
\]
This condition protects us from having a wrong mode specification, for example,
where one can use linear memory more than once via unrestricted memory region
accessing it.

For the following sections, we assume that there is a global fixed mode specification.
However, as long as the valid mode specification is provided, the system and its
metaproperties work as is.

\subsection{Syntax}
\elevator{} extends the polymorphic lambda-calculus with storing and
loading terms into a memory region (written as $ \tmStore{e}$ and
$\tmLoad{x}{e}{f}$) and with suspending terms and resuming their
evaluation. We write \(\tmCtxSusp{\hat\Psi}{e}\) for suspending
the term $e$ with the free variables in $\hat\Psi$. We simply use an
\emph{erased context} \(\hat\Psi\), which is a list of variable names,
as we do not need to track type/kind information there. 
To resume the evaluation of a suspended term (template) we use 
$\tmForceWith{e}{\sigma} $. Here the term $e$ evaluates to some
template of $e'$ with free variables in $\hat\Psi$
(i.e.\ \(\tmCtxSusp{\hat\Psi}{e'}\)). The associated substitution $\sigma$
provides a map between the variables in $\hat\Psi$ to the current context
and hence allows us to continue the execution of $e'$ with the environment
provided by $\sigma$.

\[
  \begin{array}{llcl}
   \text{Modes} & k,l,m, \andothers & &\\[0.25em]
   \text{Kinds} & K_k, \andothers & \syntaxDef & \kiType{k} \syntaxAlt \kiCtxUp{k}{l}{\Psi}{K_l}
\\[0.25em]
   \text{Normal Types} & A_k,B_k, \andothers & \syntaxDef & \tyOne{k} \syntaxAlt P_k\\
                & & \syntaxAlt & \tyCtxThunk{\hat\Psi}{A_l} \syntaxAlt \tyCtxUp{k}{l}{\Psi}{A_l} \syntaxAlt \tyDown{n}{k}{A_n}\\
                & & \syntaxAlt & \tyForall{\alpha}{K_m}{A_k} \syntaxAlt \tyLarr{A_k}{B_k}\\
    \text{Neutral Types} & P_k,Q_k, \andothers & \syntaxDef & \alpha \syntaxAlt \tyForceWith{P_m}{\sigma}
\\[0.25em]
    \text{Terms} & e,f,g, \andothers & \syntaxDef & x \syntaxAlt \tmOne\\
                & & \syntaxAlt & \tmCtxSusp{\hat\Psi}{e} \syntaxAlt \tmForceWith{e}{\sigma} \syntaxAlt \tmStore{e} \syntaxAlt \tmLoad{x}{e}{f}\\
                & & \syntaxAlt & \tmTLam{\alpha}{K_k}{e} \syntaxAlt \tmTApp{e}{A_k} \syntaxAlt \tmLam{x}{A_k}{e} \syntaxAlt \tmApp{e}{f}
\\[0.25em]
    \text{Contexts} & \Gamma,\Delta,\Psi, \andothers & \syntaxDef & \ctxEmpty \syntaxAlt \ctxCons{\Gamma}{\alpha{:}K_k} \syntaxAlt \ctxCons{\Gamma}{x{:}A_k}
\\[0.25em]
    \text{Substitutions} & \sigma,\tau,\rho, \andothers & \syntaxDef & \subEmpty \syntaxAlt \subCons{\sigma}{\alpha{\mapsto}A_k} \syntaxAlt \subCons{\sigma}{x{\mapsto_k}e}\\
  \end{array}
\]

Our types include polymorphic functions
($\tyForall{\alpha}{K_m}{A_k}$) and substructural functions,
written as $\tyLarr{A_k}{B_k}$. We also include the two adjoint modalities:
contextual up-shift, written as $\tyCtxUp{k}{l}{\Psi}{A_l}$ and
down-shift as $\tyDown{n}{k}{A_n}$.
We note that we have a mode associated with every type, in particular,
the base type $\tyOne{k}$ exists for each mode $k$.

As we discussed in the earlier examples, types and in particular type
variables may occur in different regions. This leads us to introduce
\emph{contextual kinds}, written as $\kiCtxUp{k}{l}{\Psi}{K_l}$. These
characterize suspended types, written as $\tyCtxThunk{\hat\Psi}{A_l}
$. We can use suspended types using $\tyForceWith{P_m}{\sigma}$
where the substitution \(\sigma\) instantiates the context for the
type template.

We only characterize normal and neutral types in our grammar. 
This avoids any type-level computation due to suspending and splicing
types as in the following example: 
\(\tyForceWith{\tyCtxThunk{\alpha}{\alpha}}{(\alpha{\mapsto}\tyOne{k})}\).

Although we provide only modal types and function types, the usual
linear types, such as \(\tyTensorSymb\), \(\tyPlusSymb\), and \(\tyAndSymb\)
can be easily added to the normal types of this system by
following~\citet{Jang:FSCD24}.


\subsection{Context Operations}
In preparation for our type system, we define a context split
operation $\ctxMerge{\Gamma_1}{\Gamma_2}$.
This extends the context split operation given in \citet{Jang:FSCD24}
to contexts containing type variables as well as term variables.
In particular, we allow duplicating type variables disregarding their modes
as in~\citet{Atkey:LICS18}.
As explained in the introduction, this is sensible since the type-checking
and type/kind construction are compile-time steps rather than runtime computation.

\[
  \begin{array}{l@{~\ctxMerge{}{}~}lcll}
    (\ctxCons{\Gamma_1}{\alpha{:}K_k}) & (\ctxCons{\Gamma_2}{\alpha{:}K_k}) & = & \ctxCons{(\ctxMerge{\Gamma_1}{\Gamma_2})}{\alpha{:}K_k}\\
    (\ctxCons{\Gamma_1}{\alpha{:}K_k}) & \Gamma_2 & = & \ctxCons{(\ctxMerge{\Gamma_1}{\Gamma_2})}{\alpha{:}K_k}\\
    \Gamma_1 & (\ctxCons{\Gamma_2}{\alpha{:}K_k}) & = & \ctxCons{(\ctxMerge{\Gamma_1}{\Gamma_2})}{\alpha{:}K_k}\\
    (\ctxCons{\Gamma_1}{x{:}A_k}) & (\ctxCons{\Gamma_2}{x{:}A_k}) & = & \ctxCons{(\ctxMerge{\Gamma_1}{\Gamma_2})}{x{:}A_k} &\text{if } \modeCo \in \modeSig{k}\\
    (\ctxCons{\Gamma_1}{x{:}A_k}) & \Gamma_2 & = & \ctxCons{(\ctxMerge{\Gamma_1}{\Gamma_2})}{x{:}A_k}\\
    \Gamma_1 & (\ctxCons{\Gamma_2}{x{:}A_k}) & = & \ctxCons{(\ctxMerge{\Gamma_1}{\Gamma_2})}{x{:}A_k}\\
    \ctxEmpty& \ctxEmpty & = & \ctxEmpty\\
  \end{array}
\]

For a term variable $x{:}A_k$ where $k$ does not allow contraction,
we must make sure that $x$ is used only once. However, if a mode allows
contraction, we may want to use it several times. We split a
context into two while respecting this in the definition of the
context split.

We prove that the context split is commutative and associative.

\begin{lemma}[Properties of Context Split]\mbox{}
  \begin{itemize}
  \item \emph{Commutative :} \(\ctxMerge{\Gamma_1}{\Gamma_2} = \Gamma\) if and only if
    \(\ctxMerge{\Gamma_2}{\Gamma_1} = \Gamma\)
  \item \emph{Associative :} \(\ctxMerge{(\ctxMerge{\Gamma_1}{\Gamma_2})}{\Gamma_3} = \Gamma\) if and only if
  \(\ctxMerge{\Gamma_1}{(\ctxMerge{\Gamma_2}{\Gamma_3})} = \Gamma\)
  \end{itemize}
\end{lemma}
\begin{proof}
  By induction on \(\Gamma\) and case analysis on the definition of \(\ctxMerge{}{}\).
\end{proof}

\subsection{Substitution}
As we only use normal types in \elevator, we want to define
the (parallel) substitution operation so that it preserves
the normal forms of types.
Thus, we define the hereditary substitution~\cite{Watkins02tr,Nanevski:TOCL08}
for types, which reduces type-level redices
(i.e.\ \(\tyForceWith{\tyCtxThunk{\hat\Psi}{A_l}}{\sigma}\) form)
occuring during the substitution. However, as this may increase the size of
the type to which we apply a substitution, we cannot use the size as
a termination measure for this operation. To give a termination measure for
this substitution, we use dependency-free context and kind:
%
\[
  \begin{array}{llcl}
    \text{Dependency-Free Kinds} & \kappa_k, \andothers & \syntaxDef & \kiType{k} \syntaxAlt \kiCtxUp{k}{l}{\psi}{\kappa_l}
\\
    \text{Dependency-Free Contexts} & \psi, \gamma, \andothers & \syntaxDef & \ctxEmpty \syntaxAlt \ctxCons{\gamma}{\alpha{:}\kappa_k} \syntaxAlt \ctxCons{\gamma}{x}
  \end{array}
\]
These dependency-free contexts/kinds give us the appropriate termination measure, e.g.\ they
structurally decrease when we define the substitution. We obtain these
dependency-free contexts/kinds from the usual contexts/kinds using a dependency erasure
\(\eraseForSub{\_}\) defined below.

\[
  \begin{array}{lcllcl}
\multicolumn{3}{l}{\mbox{Context erasure}: \eraseForSub{\Psi} = \psi } &
\multicolumn{3}{l}{\mbox{Kind erasure}: \eraseForSub{K_l} = \kappa_l } \\
    \eraseForSub{\ctxEmpty} & = & \ctxEmpty &
            \eraseForSub{\kiType{k}} & = & \kiType{k}
\\
    \eraseForSub{\ctxCons{\Gamma}{\alpha{:}K_k}} & = & \ctxCons{\eraseForSub{\Gamma}}{\alpha{:}\eraseForSub{K_k}}\qquad\qquad&
            \eraseForSub{\kiCtxUp{k}{l}{\Psi}{K_l}} & = & \kiCtxUp{k}{l}{\eraseForSub{\Psi}}{\eraseForSub{K_l}}
\\
    \eraseForSub{\ctxCons{\Gamma}{x{:}A_k}} & = & \ctxCons{\eraseForSub{\Gamma}}{x} & & & 
  \end{array}
\]

Now, we can define the following substitution for normal and neutral types.
\[
\newcommand{\jolt}{0em}
\newcommand{\joltp}{0.25em}
  \begin{array}{lcll}
\multicolumn{4}{l}{\mbox{Substitution for normal types:} \subNorm{\sigma}{\gamma}{A_k} = B_k}\\[0.5em]
    \subNorm{\sigma}{\gamma}{(\tyCtxUp{k}{l}{\Psi}{A_l})} & = & \tyCtxUp{k}{l}{\Psi'}{B_l} 
& \mbox{where}~\subNorm{\sigma}{\gamma}{\Psi} = \Psi' ~~\mbox{and}~~ \subNorm{\sigma}{\gamma}{A_l} = B_l\\[\jolt]
    \subNorm{\sigma}{\gamma}{(\tyDown{m}{k}{A_m})} & = & \tyDown{m}{k}{(B_m)}
& \mbox{where}~\subNorm{\sigma}{\gamma}{A_m} = B_m\\[\jolt]
    \subNorm{\sigma}{\gamma}{(\tyCtxThunk{\hat\Psi}{A_l})} & = & \tyCtxThunk{\hat\Psi}{B_l} 
& \mbox{where}~\subNorm{\sigma}{\gamma}{A_l} = B_l \\[\jolt]
    \subNorm{\sigma}{\gamma}{P_k} & = & A_k & \text{if } \subNeut{\sigma}{\gamma}{P_k} = A_k : \kappa_k\\[\jolt]
    \subNorm{\sigma}{\gamma}{P_k} & = & Q_k & \text{if } \subNeut{\sigma}{\gamma}{P_k} = Q_k\\[1em]
\multicolumn{4}{l}{\mbox{Substitution for neutral types:} \subNeut{\sigma}{\gamma}{P_k} = Q_k}\\[0.5em]
    \subNeut{\sigma}{\gamma}{\alpha} & = & A_k : \kappa_k &
\text{if } \alpha{\mapsto}A_k \in \sigma \text{ and } \alpha{:}\kappa_k \in \gamma\\[\joltp]
    \subNeut{\sigma}{\gamma}{\alpha} & = & \alpha & \text{otherwise}\\[\jolt]
    \subNeut{\sigma}{\gamma}{(\tyForceWith{P_m}{\tau})} & = & A'_k : \kappa_k 
& \text{if } \subNeut{\sigma}{\gamma}{P_m} = \tyCtxThunk{\hat\Psi}{A_k} : \kiCtxUp{m}{l}{\psi}{\kappa_k}\\
& & & \text{and}~\subNorm{\sigma}{\gamma}{\tau} = \tau'~~\text{and}~~ \subNorm{\tau'}{\psi}{A_k} = A'_k\\[\joltp]
    \subNeut{\sigma}{\gamma}{(\tyForceWith{P_m}{\tau})} & = & \tyForceWith{P'_m}{\tau'} & \text{if } \subNeut{\sigma}{\gamma}{P_m} = P'_m
~~\text{and}~~ \subNorm{\sigma}{\gamma}{\tau} = \tau'
  \end{array}
\]

In this definition, if \(\subNeut{\sigma}{\gamma}{P_k} = A_k : \kappa_k\), then
\(\kappa_k < \gamma\) holds, i.e. \(\kappa_k\) is structurally smaller than \(\gamma\). Thus, if $\subNeut{\sigma}{\gamma}{P_m} = \tyCtxThunk{\hat\Psi}{A_k} : \kiCtxUp{m}{l}{\psi}{\kappa_k}$,
we get that $\kiCtxUp{m}{l}{\psi}{\kappa_k} < \gamma$ and hence $\psi < \gamma$.
Therefore, $ \subNeut{\sigma}{\gamma}{(\tyForceWith{P_m}{\tau})}$ where we 
continue to apply the substitution to the result of $\subNeut{\sigma}{\psi}{P_m}$ is fine.

We extend this substitution operation to the context $\Psi$ by
applying $\sigma$ to each declaration in $\Psi$.
When a substitution \(\sigma\) is applied to a kind, it affects only the
\(\kiCtxUp{k}{l}{\Psi}{K_l}\) case where we apply \(\sigma\) recursively to
\(\Psi\) and \(K_l\). We also extend the substitution operation to another
substitution $\tau$, e.g. \(\subNorm{\sigma}{\gamma}{\tau}\), by applying
\(\sigma\) to each entry in $\tau$. This relies on the application of $\sigma$
to a term, which we define below. 
\[
  \begin{array}{lcll}
\multicolumn{4}{l}{\mbox{Substitution for terms:} \subNeut{\sigma}{\gamma}{e} = e'}\\[0.5em]
    \subNorm{\sigma}{\gamma}{x} & = & e & \text{if } x{\mapsto_k}e \in \sigma\\
    \subNorm{\sigma}{\gamma}{x} & = & x & \text{otherwise}\\
    \subNorm{\sigma}{\gamma}{\tmCtxSusp{\hat\Psi}{e}} & = & \tmCtxSusp{\hat\Psi}{e'} &
\mbox{where}~\subNorm{\sigma}{\gamma}{e} = e'\\
    \subNorm{\sigma}{\gamma}{\tmForceWith{e}{\tau}} & = & \tmForceWith{e'}{\tau'} &
\mbox{where}~\subNorm{\sigma}{\gamma}{e} = e'
~\mbox{and}~\subNorm{\sigma}{\gamma}{\tau} = \tau'
\\
    \subNorm{\sigma}{\gamma}{\tmStore{e}} & = & \tmStore{e'} &
\mbox{where}~\subNorm{\sigma}{\gamma}{e} = e'
\\
    \subNorm{\sigma}{\gamma}{\tmLoad{x}{e}{f}} & = & \tmLoad{x}{e'}{f'} &
\mbox{where}~\subNorm{\sigma}{\gamma}{e}=e' ~\mbox{and}~\subNorm{\sigma}{\gamma}{f} = f'
  \end{array}
\]

Both of these substitutions may fail for arbitrary input terms.
For example, in the definition of \(\subNeut{\sigma}{\gamma}{(\tyForceWith{P_m}{\tau})}\),
if neither of the conditions is true, it fails. Nonetheless, these
substitutions terminate, as shown in the next theorem. Later, in~\Cref{lem:substitution},
we prove that a substitution succeeds and gives well-typed result for well-typed
substitution and term.
\begin{theorem}[Termination of Substitution]\mbox{}
  \begin{enumerate}
  \item \(\subNorm{\sigma}{\gamma}{A_k}\) terminates, i.e.\ it either produces a result or fails.
  \item \(\subNorm{\sigma}{\gamma}{e}\) terminates, i.e.\ it either produces a result or fails.
  \end{enumerate}
\end{theorem}
\begin{proof}
  By lexicographical induction on \(\gamma\) and the target of substitution (\(A_k\) or \(e\)).
  Note that whenever we change \(\gamma\) to \(\psi\), we have \(\psi\) is a part of \(\gamma\) and $\psi < \gamma$.
\end{proof}

The single substitution \(\subNorm{\alpha{\mapsto}A_m}{\alpha{:}\kappa_m}{}\) is
defined as a special case of the above definitions.

In addition to the actual substitution operation, we need a substitution split operation,
as we need to find a substitution with matching domain for the result of the context split
when we prove the substitution lemma.
\[
  \begin{array}{l@{~\subMerge{}{}~}lcll}
    (\subCons{\sigma_1}{\alpha{\mapsto}A_k})& (\subCons{\sigma_2}{\alpha{\mapsto}A_k}) & = & \subCons{(\subMerge{\sigma_1}{\sigma_2})}{\alpha{\mapsto}A_k}\\
    (\subCons{\sigma_1}{\alpha{\mapsto}A_k}) & \sigma_2 & = & \subCons{(\subMerge{\sigma_1}{\sigma_2})}{\alpha{\mapsto}A_k}\\
    \sigma_1 & (\subCons{\sigma_2}{\alpha{\mapsto}A_k}) & = & \subCons{(\subMerge{\sigma_1}{\sigma_2})}{\alpha{\mapsto}A_k}\\
    (\subCons{\sigma_1}{x{\mapsto_k}e}) & (\subCons{\sigma_2}{x{\mapsto_k}e}) & = & \subCons{(\subMerge{\sigma_1}{\sigma_2})}{x{\mapsto_k}e} &\text{if } \modeCo \in \modeSig{k}\\
    (\subCons{\sigma_1}{x{\mapsto_k}e}) & \sigma_2 & = & \subCons{(\subMerge{\sigma_1}{\sigma_2})}{x{\mapsto_k}e}\\
    \sigma_1 & (\subCons{\sigma_2}{x{\mapsto_k}e}) & = & \subCons{(\subMerge{\sigma_1}{\sigma_2})}{x{\mapsto_k}e}\\
    \subEmpty & \subEmpty & = & \subEmpty\\
  \end{array}
\]
As in the context split operation, we allow duplication of types disregarding
the mode information and duplication of terms only for modes that allow
contraction (\(\modeCo\)).

We prove that the substitution split is commutative and associative, parallel to
the properties of the context split.

\begin{lemma}[Properties of Substitution Split]\mbox{}
  \begin{itemize}
  \item \emph{Commutative :} \(\subMerge{\sigma_1}{\sigma_2} = \sigma\) if and only if
    \(\subMerge{\sigma_2}{\sigma_1} = \sigma\)
  \item \emph{Associative :} \(\subMerge{(\subMerge{\sigma_1}{\sigma_2})}{\sigma_3} = \sigma\) if and only if
  \(\subMerge{\sigma_1}{(\subMerge{\sigma_2}{\sigma_3})} = \sigma\)
  \end{itemize}
\end{lemma}
\begin{proof}
  By induction on \(\sigma\) and case analysis on the definition of \(\subMerge{}{}\).
\end{proof}

\subsection{Type System and Substructurality}
The type system for \elevator{} uses the context split operation and substitution operation defined in the previous section.
We define when kinds, types, and contexts are well-formed in \Cref{fig:wfkinds}. We need
these well-formedness judgements since \elevator{} has contextual kinds and
type variables. 
In the following definitions, we use \(\Gamma \ge k\) (or \(k > \Gamma\)) for
the condition where every assumption in \(\Gamma\) has a mode greater than or
equal to \(k\) (or a mode less than \(k\)). Note that this condition trivially holds
when \(\Gamma\) is the empty context.
%
We presuppose the following for the well-formedness judgements:
\begin{itemize}
\item For \(\Gamma \vdash K_k \kindCls\),  \(\vdash \Gamma \ctxCls\) and \(\Gamma \ge k\)
\item For \(\Gamma \vdash A_k : K_k\),  \(\vdash \Gamma \ctxCls\) and \(\Gamma \ge k\) and \(\Gamma \vdash K_k \kindCls\)
\end{itemize}

\begin{figure}
  \centering
\[
  \begin{array}{c}
    \multicolumn{1}{l}{\fbox{\(\Gamma \vdash K_k \kindCls\) when kind \(K_k\) is well-formed under context \(\Gamma\)}}\\[1em]
 \infer
    {\Gamma \vdash \kiType{k} \kindCls}
    {}
    \qquad
	\infer
    {\Gamma \vdash \kiCtxUp{k}{l}{\Psi}{K_l} \kindCls}
    {k > \Psi \ge l
    \qquad \vdash \ctxCons{\Gamma}{\Psi} \ctxCls
    \qquad \ctxCons{\Gamma}{\Psi} \vdash K_l \kindCls}
    \\[1em]
    \multicolumn{1}{l}{\fbox{\(\Gamma \vdash A_k : K_k\) when type \(A_k\) is a well-formed type of kind \(K_k\) under context \(\Gamma\)}}\\[1em]
	\infer
    {\Gamma \vdash \tmCtxSusp{\hat\Psi}{A_l} : \kiCtxUp{k}{l}{\Psi}{K_l}}
    {\ctxCons{\Gamma}{\Psi} \vdash A_l : K_l}
    \\[0.75em]
	\infer
    {\ctxMerge{\Gamma_1}{\Gamma_2} \vdash \tyForceWith{A_m}{\sigma} : \subNorm{\sigma}{\eraseForSub{\Psi}}{K_k}}
    {
    \begin{array}{ll}
      \Gamma_1 \ge m \quad \vdash \Gamma_1 \ctxCls
      & \vdash \ctxCons{(\ctxMerge{\Gamma_1}{\Gamma_2})}{\Psi} \ctxCls\\
      \Gamma_1 \vdash \tyCtxUp{m}{k}{\Psi}{K_k} \kindCls
      \quad \Gamma_1 \vdash A_m : \tyCtxUp{m}{k}{\Psi}{K_k}
      &  \ctxMerge{\Gamma_1}{\Gamma_2} \vdash \sigma : \Psi
    \end{array}}
    \\[0.75em]
	\infer
    {\Gamma \vdash \tyCtxUp{k}{l}{\Psi}{A_l} : \kiType{k}}
    {k > \Psi \ge l
    & \vdash \ctxCons{\Gamma}{\Psi} \ctxCls
    & \ctxCons{\Gamma}{\Psi} \vdash A_l : \kiType{l}}
    \qquad
	\infer
    {\ctxMerge{\Gamma_1}{\Gamma_2} \vdash \tyDown{m}{k}{A_m} : \kiType{k}}
    {\Gamma_1 \ge m \ge k
    &\vdash \Gamma_1 \ctxCls
    &\Gamma_1 \vdash A_m : \kiType{m}}
    \\[0.75em]
	\infer
    {\ctxMerge{\Gamma_1}{\Gamma_2} \vdash \tyForall{\alpha}{K_m}{A_k} : \kiType{k}}
    {\Gamma_1 \ge m \ge k
    &\vdash \Gamma_1 \ctxCls
    &\Gamma_1 \vdash K_m \kindCls
    &\ctxCons{(\ctxMerge{\Gamma_1}{\Gamma_2})}{\alpha{:}K_m} \vdash A_k : \kiType{k}}
    \\[0.75em]
    \infer
    {\Gamma \vdash \tyArr{A_k}{B_k} : \kiType{k}}
    {\Gamma \vdash A_k : \kiType{k}
    \qquad \Gamma \vdash B_k : \kiType{k}}
    \qquad
    \infer
    {\Gamma \vdash \tyOne{k} : \kiType{k}}
    {}
    \\[0.75em]
    \multicolumn{1}{l}{\fbox{\(\vdash \Gamma \ctxCls\) when \(\Gamma\) is a well-formed context}}\\[0.75em]
	\infer
    {\vdash \ctxEmpty \ctxCls}
    {}
    \qquad
	\infer
    {\vdash \ctxCons{(\ctxMerge{\Gamma_1}{\Gamma_2})}{\alpha{:}K_k} \ctxCls}
    {\Gamma_1 \ge k
    \qquad \vdash \Gamma_1 \ctxCls
    \qquad \vdash \ctxMerge{\Gamma_1}{\Gamma_2} \ctxCls
    \qquad \Gamma_1 \vdash K_k \kindCls}
\\[0.75em]
	\infer
    {\vdash \ctxCons{(\ctxMerge{\Gamma_1}{\Gamma_2})}{x{:}A_k} \ctxCls}
    {\Gamma_1 \ge k
    \qquad \vdash \Gamma_1 \ctxCls
    \qquad \vdash \ctxMerge{\Gamma_1}{\Gamma_2} \ctxCls
    \qquad \Gamma_1 \vdash A_k : \kiType{k}}
  \end{array}
\]  
  \caption{Well-Formedness of Kinds, Types, and Contexts}
  \label{fig:wfkinds}
\end{figure}

In each rule, we keep enough premises so that we propagate these
presuppositions. For example, in the following rule
for the well-formedness of \(\kiCtxUp{k}{l}{\Psi}{K_l}\),
\[
  \infer
  {\Gamma \vdash \kiCtxUp{k}{l}{\Psi}{K_l} \kindCls}
  {k > \Psi \ge l
    \qquad \vdash \ctxCons{\Gamma}{\Psi} \ctxCls
    \qquad \ctxCons{\Gamma}{\Psi} \vdash K_l \kindCls}
\]
we have \(\Psi \ge l\) and \(\vdash \ctxCons{\Gamma}{\Psi} \ctxCls\) as
premises since we presuppose that for \(\ctxCons{\Gamma}{\Psi} \vdash K_l \kindCls\).
This rule also has \(k > \Psi\) as a premise. This premise makes sure that
everything in a template, including the context \(\Psi\), can be treated
syntactically for mode \(k\). One important remark on \(k > \Psi \ge l\) is
about when \(\Psi\) is the empty context. If \(\Psi\) is empty, even though
\(k > \Psi\) and \(\Psi \ge l\) separately trivially holds, we still require
\(k \ge l\). This makes sure that, when the context is empty,
our contextual up-shift is compatible with the non-contextual
version~\cite{Jang:FSCD24}, which allows up-shift only from a lower mode
to a higher mode.

There are also rules where we can omit some premises because of the presuppositions. 
The following rule is such an example.
\[
  \infer
  {\Gamma \vdash \tmCtxSusp{\hat\Psi}{A_l} : \kiCtxUp{k}{l}{\Psi}{K_l}}
  {\ctxCons{\Gamma}{\Psi} \vdash A_l : K_l}
\]
Here, we do not check \(k > \Psi \ge l\) nor \(\vdash \Gamma,\Psi \ctxCls\),
as we can get them by inversion on the presupposed
\(\Gamma \vdash \kiCtxUp{k}{l}{\Psi}{K_l} \kindCls\).

The case of \(\tyForceWith{A_m}{\sigma}\) shows how we use context split operation
to provide a correct context.
\[
	\infer
    {\ctxMerge{\Gamma_1}{\Gamma_2} \vdash \tyForceWith{A_m}{\sigma} : \subNorm{\sigma}{\eraseForSub{\Psi}}{K_k}}
    {
    \begin{array}{ll}
      \Gamma_1 \ge m \quad \vdash \Gamma_1 \ctxCls
      & \vdash \ctxCons{(\ctxMerge{\Gamma_1}{\Gamma_2})}{\Psi} \ctxCls\\
      \Gamma_1 \vdash A_m : \tyCtxUp{m}{k}{\Psi}{K_k} \quad
      \Gamma_1 \vdash \tyCtxUp{m}{k}{\Psi}{K_k} \kindCls
      & \ctxMerge{\Gamma_1}{\Gamma_2} \vdash \sigma : \Psi
    \end{array}}
\]
In this rule, we need \(\Gamma_1 \ge m\) and \(\vdash \Gamma_1 \ctxCls\) for
\(\Gamma_1 \vdash A_m : \tyCtxUp{m}{k}{\Psi}{K_k}\). To obtain such a context,
we use the context split operation (\(\ctxMerge{\Gamma_1}{\Gamma_2}\)). However,
we never use \(\Gamma_2\) alone in this rule, and thus we do not have
\(\vdash \Gamma_2 \ctxCls\) as a premise. 
In the above rule, the context for \(\sigma\) maps variables from $\Psi$ to \(\ctxMerge{\Gamma_1}{\Gamma_2}\) even though \(A_m\) and $K_k$ already use variables in \(\Gamma_1\). Hence, type variables might, for example, be occurring more than once in $\subNorm{\sigma}{\eraseForSub{\Psi}}{K_k}$. 
We allow this because types and kinds are compile-time concepts,
and therefore we never consider an assumption being consumed in types.

The rule for type abstraction \(\tyForall{\alpha}{K_m}{A_k}\) is also interesting.
Note how it abstracts over a type variable of a kind in a higher mode \(m \ge k\).
This general abstraction gives us the power to use type templates across
multiple modes. In fact, without this general abstraction, we cannot form
polymorphic program generators that require cross-stage persistent type variables
as in \Cref{subsec:motivation-poly-code-gen}.


We describe the type system for terms in \Cref{fig:exptyping}. In the typing rules,  \(\Gamma_W\)  represents a context all of whose entries are
in modes with weakening. In other words, it only comprises assumptions that we
do not need to use.

\begin{figure}
  \centering
\[
\small
  \begin{array}{c}
    \multicolumn{1}{l}{\fbox{\(\Gamma \vdash e : A_k\) when term \(e\) is of type \(A_k\) under context \(\Gamma\)}}\\[1em]
	\infer
    {\ctxMerge{x:A_k}{\Gamma_W} \vdash x : A_k}
    {}
    \qquad
	\infer
    {\Gamma \vdash \tmCtxSusp{\hat\Psi}{e} : \tyCtxUp{k}{l}{\Psi}{A_l}}
    {\ctxCons{\Gamma}{\Psi} \vdash e : A_l}
    \qquad
	\infer
    {\ctxMerge{\Gamma}{\Gamma_W} \vdash \tmStore{e} : \tyDown{m}{k}{A_m}}
    {\Gamma \ge m
    \qquad \vdash \Gamma \ctxCls
    \qquad \Gamma \vdash e : A_m}
    \\[0.75em]
	\infer
    {\Gamma \vdash \tmTLam{\alpha}{K_m}{e} : \tyForall{\alpha}{K_m}{A_k}}
    {\ctxCons{\Gamma}{\alpha{:}K_m} \vdash e : A_k}
    \qquad
	\infer
    {\Gamma \vdash \tmLam{x}{A_k}{e} : \tyArr{A_k}{B_k}}
    {\ctxCons{\Gamma}{x{:}A_k} \vdash e : B_k}
    \qquad
    \infer
    {\Gamma_W \vdash \tmOne : \tyOne{k}}
    {}
    \\[0.75em]
	\infer
    {\ctxMerge{\Gamma_1}{\Gamma_2} \vdash \tmForceWith{e}{\sigma} : \subNorm{\sigma}{\eraseForSub{\Psi}}{A_k}}
    {
    \begin{array}{ll}
      \vdash \Gamma_1 \ctxCls \quad \Gamma_1 \ge m
      & \vdash \ctxCons{\Gamma_2}{\Psi} \ctxCls\\
      \Gamma_1 \vdash \tyCtxUp{m}{k}{\Psi}{A_k} : \kiType{m}
      \quad \Gamma_1 \vdash e : \tyCtxUp{m}{k}{\Psi}{A_k}
      & \Gamma_2 \vdash \sigma : \Psi
    \end{array}
}
    \\[0.75em]
	\infer
    {\ctxMerge{\Gamma_1}{\Gamma_2} \vdash \tmLoad{x}{e}{f} : B_k}
    {
    \begin{array}{ll}
      \vdash \Gamma_1 \ctxCls\qquad     \Gamma_1 \ge n \ge k
      & \vdash \ctxCons{\Gamma_2}{x{:}A_m} \ctxCls\\
      \Gamma_1 \vdash \tyDown{m}{n}{A_m} : \kiType{m}
      \qquad \Gamma_1 \vdash e : \tyDown{m}{n}{A_m}
      & \ctxCons{\Gamma_2}{x{:}A_m} \vdash f : B_k
    \end{array}
    }
    \\[0.75em]
	\infer
    {\ctxMerge{\Gamma_1}{\Gamma_2} \vdash \tmTApp{e}{A_m} : \subNorm{\alpha{\mapsto}A_m}{\alpha:\eraseForSub{K_m}}{B_k}}
    {
    \begin{array}{ll}
      & \vdash \Gamma_1 \ctxCls \qquad \Gamma_1 \ge m \ge k\\
      \ctxMerge{\Gamma_1}{\Gamma_2} \vdash \tyForall{\alpha}{K_m}{B_k} : \kiType{k}
      \qquad \ctxMerge{\Gamma_1}{\Gamma_2} \vdash e : \tyForall{\alpha}{K_m}{B_k}
      & \Gamma_1 \vdash K_m \kindCls
        \qquad \Gamma_1 \vdash A_m : K_m
    \end{array}
    }
    \\[0.75em]
	\infer
    {\ctxMerge{\Gamma_1}{\Gamma_2} \vdash \tmApp{e}{f} : B_k}
    {
    \begin{array}{ll}
      \vdash \Gamma_1 \ctxCls
      & \vdash \Gamma_2 \ctxCls
      \\
      \Gamma_1 \vdash \tyArr{A_k}{B_k} : \kiType{k}
      \quad \Gamma_1 \vdash e : \tyArr{A_k}{B_k}
      & \Gamma_2 \vdash A_k : \kiType{k}
        \quad \Gamma_2 \vdash f : A_k
    \end{array}
    }
    \\[1em]
    \multicolumn{1}{l}{\fbox{\(\Gamma \vdash \sigma : \Psi\) when \(\sigma\) is a well-typed substitution from \(\Psi\) to \(\Gamma\)}}\\[1em]
	\infer
    {\Gamma_W \vdash \subEmpty : \ctxEmpty}
    {}
    \qquad
	\infer
    {\ctxMerge{\Gamma_1}{\Gamma_2} \vdash \subCons{\sigma}{\alpha{\mapsto}A_k} : \ctxCons{\Psi}{\alpha{:}K_k}}
    {\Gamma_1 \ge k
    \qquad \ctxMerge{\Gamma_1}{\Gamma_2} \vdash \sigma : \Psi
    \qquad \vdash \Gamma_1 \ctxCls
    \qquad \Gamma_1 \vdash \subNorm{\sigma}{\eraseForSub{\Psi}}{K_k} \kindCls
    \qquad \Gamma_1 \vdash A_k : \subNorm{\sigma}{\eraseForSub{\Psi}}{K_k}}
    \\[0.75em]
	\infer
    {\ctxMerge{\Gamma_1}{\Gamma_2}\vdash \subCons{\sigma}{x{\mapsto_k}e} : \ctxCons{\Psi}{x{:}A_k}}
    {\Gamma_1 \ge k
    \qquad \vdash  \Gamma_2 \ctxCls
    \qquad \Gamma_2 \vdash \sigma : \Psi
    \qquad \vdash  \Gamma_1 \ctxCls
    \qquad \Gamma_1 \vdash \subNorm{\sigma}{\eraseForSub{\Psi}}{A_k} : \kiType{k}
    \qquad \Gamma_1 \vdash e : \subNorm{\sigma}{\eraseForSub{\Psi}}{A_k}}
  \end{array}
\]
  
  \caption{Typing Rules for Terms and Substitutions}
  \label{fig:exptyping}
\end{figure}

For the typing judgements, we presuppose the following:
\begin{itemize}
\item For \(\Gamma \vdash e : A_k\),  \(\vdash \Gamma \ctxCls\) and \(\Gamma \ge k\) and \(\Gamma \vdash A_k : \kiType{k}\)
\item For \(\Gamma \vdash \sigma : \Delta\),  \(\vdash \ctxCons{\Gamma}{\Delta} \ctxCls\)
\end{itemize}

As in the well-formedness case, we have premises to propagate the presuppositions.
For instance, in the typing rule for \(\tmLoad{x}{e}{f}\),
\[
  \infer
  {\ctxMerge{\Gamma_1}{\Gamma_2} \vdash \tmLoad{x}{e}{f} : B_k}
  {
    \begin{array}{ll}
      \vdash \Gamma_1 \ctxCls\qquad     \Gamma_1 \ge n \ge k
      & \vdash \ctxCons{\Gamma_2}{x{:}A_m} \ctxCls\\
      \Gamma_1 \vdash \tyDown{m}{n}{A_m} : \kiType{m}
      \qquad \Gamma_1 \vdash e : \tyDown{m}{n}{A_m}
      & \ctxCons{\Gamma_2}{x{:}A_m} \vdash f : B_k
    \end{array}
  }
\]
we have \(\Gamma_1 \ge n\), \(\vdash \Gamma_1\) as premises so that
we presuppose those for \(\Gamma_1 \vdash \tyDown{m}{n}{A_m} : \kiType{m}\)
and \(\Gamma_1 \vdash e : \tyDown{m}{n}{A_m}\). Likewise, we have
\(\vdash \ctxCons{\Gamma_2}{x{:}A_m} \ctxCls\) for
\(\ctxCons{\Gamma_2}{x{:}A_m} \vdash f : B_k\).

Note that we use both \(\Gamma_1\) and \(\Gamma_2\) separately
in the above rule. This is different from the well-formedness rules for contextual types. 
While in the latter case we do not track usage to support a notion of cross-stage persistence for type variables, for terms,
we must track usage to provide memory access guarantees. In particular, 
we should not use a term variable in a mode without contraction.
Thus, we use \(\Gamma_2\) which contains variables that can be
legitimately used in \(f\), either because those are not used in
\(e\) or because those are reusable due to their modes.

This distinction between terms and types regarding usage tracking
affects all rules where terms contain more than two subterms, and
we have to distribute our assumptions when checking terms. If
there are multiple subterms, we split the typing context.
What happens if a term contains a type as a subterm as in a type application \(\tmTApp{e}{A_m}\)?
\[
  \infer
  {\ctxMerge{\Gamma_1}{\Gamma_2} \vdash \tmTApp{e}{A_m} : \subNorm{\alpha{\mapsto}A_m}{\alpha:\eraseForSub{K_m}}{B_k}}
  {
    \begin{array}{ll}
      & \vdash \Gamma_1 \ctxCls \qquad \Gamma_1 \ge m \ge k\\
      \ctxMerge{\Gamma_1}{\Gamma_2} \vdash \tyForall{\alpha}{K_m}{B_k} : \kiType{k}
      \qquad \ctxMerge{\Gamma_1}{\Gamma_2} \vdash e : \tyForall{\alpha}{K_m}{B_k}
      & \Gamma_1 \vdash K_m \kindCls
        \qquad \Gamma_1 \vdash A_m : K_m
    \end{array}
  }
\]
Again, we do not count the usage of assumptions in \(A_m\).
Only \(e\) can use an assumption, so we keep the context as-is for
\(\ctxMerge{\Gamma_1}{\Gamma_2} \vdash e : \tyForall{\alpha}{K_m}{B_k}\).
The typing judgement for terms guarantees that the variable references follow substructural rules in their modes.
\begin{theorem}[Substructurality of Variable References]
  If \(\;\Gamma \vdash e : A_k\) and \(x{:}B_m \in \Gamma\) and \(\modeCo \not\in \modeSig{m}\), \(x\) occurs in \(e\) at most once.
  Likewise, if \(\modeWk \not\in \modeSig{m}\), \(x\) occurs in \(e\) at least once.
\end{theorem}
\begin{proof}
  By induction on the typing derivation. Note that a type, kind, or context cannot refer to a term variable.
\end{proof}

The well-formedness/typing judgements are stable under substitutions.
To prove this stability, we first need to establish some properties of the substitution split operation regarding well-typed objects.
\begin{lemma}[Splitting]\label{lem:splitting}
  For \(\Gamma \vdash \sigma : \ctxMerge{\Delta_1}{\Delta_2}\), there exists \(\Gamma_1\), \(\Gamma_2\), \(\sigma_1\) and \(\sigma_2\) such that \(\sigma = \subMerge{\sigma_1}{\sigma_2}\) and \(\Gamma = \ctxMerge{\Gamma_1}{\Gamma_2}\) and \(\Gamma_1 \vdash \sigma_1 : \Delta_1\) and  \(\Gamma_2 \vdash \sigma_2 : \Delta_2\).
\end{lemma}
\begin{proof}
  By induction on the typing derivation for \(\sigma\).
  Note that the condition on the mode signature (\(m \geq k \text{ then } \modeSig{m} \subseteq \modeSig{k}\))
  is critical here. This is because whenever \(\Delta_1\) and \(\Delta_2\) duplicate a
  variable, all variables in \(\Gamma\) that \(\sigma\) uses to instantiate the duplicated
  variable should be also duplicated.
\end{proof}

This lemma states that when we split the domain of a well-typed substitution,
we can split it into two well-typed substitutions of the split domains.
We also need the following lemma, which states that the split substitution has
the same effect as the original substitution if the target is well-typed under
the split domain.

\begin{lemma}[Effect of Splitting]\label{lem:effect-of-splitting}
  For \(\ctxMerge{\Gamma_1}{\Gamma_2} \vdash \sigma : \ctxMerge{\Delta_1}{\Delta_2}\), \(\Gamma_1 \vdash \sigma_1 : \Delta_1\), \(\sigma = \ctxMerge{\sigma_1}{\sigma_2}\), $\eraseForSub{\ctxMerge{\Delta_1}{\Delta_2}} = \delta$, and $\eraseForSub{\Delta_1} = \delta_1$,
  \begin{enumerate}
  \item If \(\vdash \Gamma_1, \Psi \ctxCls\) then \(\subNorm{\sigma}{\delta}{\Psi} = \subNorm{\sigma_1}{\delta_1}{\Psi}\)
  \item If \(\Gamma_1 \vdash K_k \kindCls\) then \(\subNorm{\sigma}{\delta}{K_k} = \subNorm{\sigma_1}{\delta_1}{K_k}\)
  \item If \(\Gamma_1 \vdash T_k : K_k\) then \(\subNorm{\sigma}{\delta}{T_k} = \subNorm{\sigma_1}{\delta_1}{T_k}\)
  \item If \(\Gamma_1 \vdash e : T_k\) then \(\subNorm{\sigma}{\delta}{e} = \subNorm{\sigma_1}{\delta_1}{e}\)
  \end{enumerate}
\end{lemma}
\begin{proof}
  By lexicographical induction on \(\delta\), \(\delta_1\), and the typing derivation of the target of the substitution.
\end{proof}

With these two lemmas, we can now prove the substitution lemma.
\begin{lemma}[Substitution Lemma]\label{lem:substitution} If \(\Gamma_2 \vdash \sigma : \Delta\) and $\eraseForSub{\Delta} = \delta$,
  \begin{enumerate}[label=(\arabic*), ref=(\arabic*)]
  \item\label{itm:substitution-context}
    If \(\vdash \ctxCons{\Gamma_1}{\ctxCons{\Delta}{\Psi}} \ctxCls\)
    then \(\vdash \ctxCons{\Gamma_1}{\ctxCons{\Gamma_2}{\subNorm{\sigma}{\delta}{(\Psi)}}} \ctxCls\)
  \item\label{itm:substitution-kind}
    If \(\ctxCons{\Gamma_1}{\ctxCons{\Delta}{\Psi}} \vdash K_k \kindCls\)
    then \(\ctxCons{\Gamma_1}{\ctxCons{\Gamma_2}{\subNorm{\sigma}{\delta}{(\Psi)}}} \vdash \subNorm{\sigma}{\delta}{(K_k)} \kindCls\)
  \item\label{itm:substitution-type}
    If \(\ctxCons{\Gamma_1}{\ctxCons{\Delta}{\Psi}} \vdash A_k : K_k\)
    then \(\ctxCons{\Gamma_1}{\ctxCons{\Gamma_2}{\subNorm{\sigma}{\delta}{(\Psi)}}} \vdash \subNorm{\sigma}{\delta}{(A_k)} : \subNorm{\sigma}{\delta}{(K_k)}\)
  \item\label{itm:substitution-term}
    If \(\ctxCons{\Gamma_1}{\ctxCons{\Delta}{\Psi}} \vdash e : A_k\)
    then \(\ctxCons{\Gamma_1}{\ctxCons{\Gamma_2}{\subNorm{\sigma}{\delta}{(\Psi)}}} \vdash \subNorm{\sigma}{\delta}{e} : \subNorm{\sigma}{\delta}{(A_k)}\)
  \item\label{itm:substitution-substitution}
    If \(\ctxCons{\Gamma_1}{\ctxCons{\Delta}{\Psi}} \vdash \tau : \Phi\)
    then \(\ctxCons{\Gamma_1}{\ctxCons{\Gamma_2}{\subNorm{\sigma}{\delta}{(\Psi)}}} \vdash \subNorm{\sigma}{\delta}{\tau} : \subNorm{\sigma}{\delta}{(\Phi)}\)
  \end{enumerate}
\end{lemma}
\begin{proof}
  By lexicographical induction on \(\delta\) and the well-formedness (or typing) derivation.
  When the typing derivation splits \(\Gamma_1\), apply \Cref{lem:splitting} and \Cref{lem:effect-of-splitting}
  as needed.
\end{proof}

\subsection{Small-step Semantics and Type Safety}
Our small-step semantics fundamentally depends on the modes of terms.
Thus, in this section, we use \(\tmCtxSuspElab{k}{l}{\hat\Psi}{e}\)
annotated with its modes instead of \(\tmCtxSusp{\hat\Psi}{e}\)
to clarify the mode in which the term is evaluated.
These modes are obtained by elaborating terms while checking
their types. Likewise, we elaborate a substitution
during the type checking and simply write 
\(\sub{\sigma}{e}\) instead of \(\subNorm{\sigma}{\delta}{e}\).

How do modes affect evaluation? To answer this, we first
need to identify the redices in \elevator{}.
First, the $\beta$-redex (\(\tmApp{(\tmLam{x}{A_k}{e})}{e'}\))
lets us invoke functions.
\(\loadStoreRed\)-redex (\(\tmLoadElab{m}{n}{x}{\tmStoreElab{m}{n}{e}}{e'}\))
replaces the pointer $x$ with the actual value $e$ in the term $e'$.
The most important one is \(\forceSuspRed\)-redex
(\(\tmForceWithElab{m}{k}{\tmCtxSuspElab{m}{k}{\hat\Gamma}{e}}{\sigma}\)).
This allows us to either splice a template into another template or
execute a template depending on the mode in which it is evaluated.
As an example of this dependence on a mode, recall the mode structure of
\Cref{sec:motivations} where we distinguish between code (mode $\mC$)
and programs (mode $\mP$) where $\mC \geq \mP$.
To arrive at a suitable metaprogramming foundation, we want to eagerly
compose program templates while keeping any other redices
irrelevant to the template composition. Note that, if a code generator in
\(\mC\) requires some evaluation such as eliminating $\beta$-redices
(at \(\mC\)), we permit such evaluation to generate a template.
However, if the template itself contains
$\beta$-redices at \(\mP\), we avoid reducing them, since
program templates describe suspended programs. These $\beta$-redices
are eventually reduced at \(\mP\) only when we force the execution
of the suspended program template.
Therefore, 
the decision to reduce a $\beta$ redex depends on whether it occurs inside a
suspended program or not. 
In general, when we
have a term with a redex at mode \(l\), whether we reduce the redex
depends on if it occurs within \(\tmCtxSuspElab{k}{n}{\hat\Psi}{-}\)
for some mode \(k\) satisfying \(l \not\geq k\) (which we will call
the \emph{ambient mode of a template}).
When the redex is part of such a template, we treat it just as
a syntactic structure, aligning with the design principle for operational
semantics outlined in the introduction:
\begin{quotation}
  The evaluation of a term should not know anything about
  a term in an inaccessible memory other than its syntactic structure.
\end{quotation}

To capture this mode-dependent behaviour, we distinguish between
(1) term evaluation and (2) template reduction. Template reduction
allows reductions for suspended programs, i.e. under
\(\tmCtxSuspElab{k}{n}{\hat\Psi}{-}\) for some ambient mode \(k\).
This distinction leads to a more fine-grained notion of values in \elevator.
In particular, we also distinguish between weak normal/neutral form
and
normal template in \(k\) (see \Cref{fig:result-syntax}).
The weak normal/neutral form is the result of term evaluation,
and the normal template in mode \(k\) is the result of the template
reduction under \(\tmCtxSuspElab{k}{n}{\hat\Psi}{-}\).

The weak normal/neutral form is weak in the sense that it does not evaluate
a term under type/term lambda abstractions nor the body of a
\(\tmLoadName\) term. One can check this in \Cref{fig:result-syntax};
\(\tmTLam{\alpha}{K_k}{e}\), \(\tmLam{x}{A_k}{e}\), and
\(\tmLoadElab{m}{k}{x}{q}{e}\) all contain a possibly-reducible term \(e\).
The more interesting case is the \(\tmCtxSuspElab{k}{l}{\hat\Psi}{a^k}\) case.
In this case, the weak normal form contains the normal template \(a^k\)
instead of a weak normal/neutral form. This is because handling evaluation
under \(\tmCtxSuspElab{k}{l}{\hat\Psi}{-}\) is the reason why we introduce
the template reduction, resulting in a normal template.

For the normal template in mode \(k\), we define two cases for each syntactic form
that can raise the mode of the subterm to an accessible mode from \(k\).
For instance, under \(l \not\ge k\), \(e\) in \(\tmStoreElab{m}{l}{e}\)
lives at mode \(m\), and since \(m \ge l\), we do not know whether
\(m \ge k\) or not. If \(m \ge k\), then \(m\) is now accessible from
the ambient mode \(k\), thus we evaluate \(e\) into its normal form.
Conversely, if \(m \not\ge k\), then \(m\) is still not accessible
from the ambient mode \(k\), leading us to reduce \(e\) as a template.
Therefore, \(\tmStoreElab{m}{l}{e}\) in the first case reduces to
\(\tmStoreElab{m}{l}{v}\) (where \(v\) is a weak normal form)
for the syntax of the normal template while the second case gives
\(\tmStoreElab{m}{l}{a^k}\) (where \(a^k\) is a normal template).


Why do we define the syntax of weak normal forms, which can be open,
instead of values? This is because a normal template may
introduce a variable that can be used in a later weak normal form.
For example, consider the following normal template in ambient mode \(k\):
\[
\tmLoadElab{m}{l}{x}{(\tmApp{(\tmLam{x}{\tyOne{l}}{\tmStoreElab{m}{l}{\tmLam{y}{\tyOne{m}}{y}}})}{\tmOne})}{\tmStoreElab{m}{l}{\tmApp{x}{\tmOne}}}
\]
where \(m > k > l\). Since \(k > l\), we do not evaluate \(\beta\)-redex
\((\tmApp{(\tmLam{x}{\tyOne{l}}{\tmStoreElab{m}{l}{\tmLam{y}{\tyOne{m}}{y}}})}{\tmOne})\)
at \(l\). However, as \(m > k\), we evaluate the inner terms at mode \(m\),
but \(\tmApp{x}{\tmOne}\) is open. Because of this, we need to
define the open evaluation result.

Nonetheless, as we will prove in this section, \elevator{} preserves typing,
so evaluation of a closed term gives another closed term. Therefore, if the evaluation
reaches a weak normal form, it must be closed.
Consequently, we define \emph{values} as closed weak normal forms.

\begin{figure}[ht]
  \[
    \begin{array}{llcll}
      \text{Weak Normal Forms} & v,u,\andothers & \syntaxDef & \tmCtxSuspElab{k}{l}{\hat\Psi}{a^k} \syntaxAlt \tmStoreElab{k}{l}{v}\\
                               & & \syntaxAlt & \tmTLam{\alpha}{K_k}{e} \syntaxAlt \tmLam{x}{A_k}{e} \syntaxAlt \tmOne \syntaxAlt q\\
      \text{Weak Neutral Forms} & q,r,\andothers & \syntaxDef & x \syntaxAlt \tmForceWithElab{m}{k}{q}{\sigma} \syntaxAlt \tmLoadElab{m}{k}{x}{q}{e}\\
                               & & \syntaxAlt & \tmTApp{q}{A_k} \syntaxAlt \tmApp{q}{v}\\[0.5em]
      \text{Normal Templates} & a^k,b^k,\andothers & \syntaxDef & x \syntaxAlt \tmCtxSuspElab{l_0}{l_1}{\hat\Psi}{a^k}\\
      \text{in mode \(k\)} & & \syntaxAlt & \tmForceWithElab{l_0}{l_1}{a^k}{\sigma} & (\text{If } l_0 \not\ge k)\\
                               & & \syntaxAlt & \tmForceWithElab{m}{l_1}{q}{\sigma} & (\text{If } m \ge k)\\
                               & & \syntaxAlt & \tmStoreElab{l_0}{l_1}{a^k} & (\text{If } l_0 \not\ge k)\\
                               & & \syntaxAlt & \tmStoreElab{m}{l_1}{v} & (\text{If } m \ge k)\\
                               & & \syntaxAlt & \tmLoadElab{l_0}{l_1}{x}{a^k}{b^k} & (\text{If } l_0 \not\ge k)\\
                               & & \syntaxAlt & \tmLoadElab{m}{l_1}{x}{v}{b^k} & (\text{If } m \ge k)\\
                               & & \syntaxAlt & \tmTLam{\alpha}{K_l}{a^k} \syntaxAlt \tmTApp{a^k}{A_l}\\
                               & & \syntaxAlt & \tmLam{x}{A_l}{a^k} \syntaxAlt \tmApp{a^k}{b^k} \syntaxAlt \tmOne\\

    \end{array}
  \]
  \caption{Syntax of weak normal/neutral forms and normal templates in mode \(k\)}
  \label{fig:result-syntax}
\end{figure}

Using these definitions, we can now give
the small-step semantics for \elevator. The small-step semantics given
in \Cref{fig:small-step} contains two judgements, \(e \stepto e'\)
(evaluation) and
\(e \refinestepto{m} e'\) (template reduction): the first one describes the stepping of
a term $e$ to a term $e'$; the second one is used to
reduce a template under an ambient mode $m$. During this
reduction, we only reduce code splices,
i.e. a $\mathsf{force}{-}\mathsf{susp}$ redex that occurs in $e$.

\begin{figure}[ht]
  \[
    \small
    \begin{array}{c}
      \multicolumn{1}{l}{\fbox{\(e \stepto e'\) term \(e\) steps to \(e'\)}}\\[1em]
      \multicolumn{1}{l}{\mbox{Reductions}}\\[0.5em]
      \infer
      {\tmForceWithElab{m}{k}{\tmCtxSuspElab{m}{k}{\hat\Psi}{a^m}}{\sigma} \stepto \sub{\sigma}{a^m}}
      {}
\quad
      \infer
      {\tmLoadElab{m}{n}{x}{\tmStoreElab{m}{n}{v}}{e_2} \stepto \sub{x{\mapsto_m}v}{e_2}}
      {}
\\[1em]
      \infer
      {\tmTApp{(\tmTLam{\alpha}{K_k}{e})}{A_k} \stepto \sub{\alpha{\mapsto}A_k}{e}}
      {}
\quad
      \infer
      {\tmApp{(\tmLam{x}{A_k}{e})}{u} \stepto \sub{x{\mapsto_k}u}{e}}
      {}
\\[1em]
\multicolumn{1}{l}{\mbox{Congruences}}\\
      \infer
      {\tmCtxSuspElab{k}{l}{\hat\Psi}{e} \stepto \tmCtxSuspElab{k}{l}{\hat\Psi}{e'}}
      {e \refinestepto{k} e'}
      \qquad
      \infer
      {\tmForceWithElab{m}{k}{e}{\sigma} \stepto \tmForceWithElab{m}{k}{e'}{\sigma}}
      {e \stepto e'}
      \qquad

      \\[1em]
      \infer
      {\tmStoreElab{m}{k}{e} \stepto \tmStoreElab{m}{k}{e'}}
      {e \stepto e'}
      \qquad
      \infer
      {\tmLoadElab{m}{n}{x}{e_1}{e_2} \stepto \tmLoadElab{m}{n}{x}{e_1'}{e_2}}
      {e_1 \stepto e_1'}
      \\[1em]
      \infer
      {\tmTApp{e}{A_k} \stepto \tmApp{e'}{A_k}}
      {e \stepto e'}
      \qquad
      \infer
      {\tmApp{e}{f} \stepto \tmApp{e'}{f}}
      {e \stepto e'}
      \qquad
      \infer
      {\tmApp{v}{f} \stepto \tmApp{v}{f'}}
      {f \stepto f'}
      \\[2em]
      \multicolumn{1}{l}{\fbox{\(e \refinestepto{m} e'\) term \(e\) is one-step refined into \(e'\) in the evaluation of mode \(m\)}}\\[1.5em]
      \infer
      {\tmCtxSuspElab{k}{l}{\hat\Gamma}{e} \refinestepto{m} \tmCtxSuspElab{k}{l}{\hat\Gamma}{e'}}
      {e \refinestepto{m} e'}
      \qquad
      \infer
      {\tmForceWithElab{n}{k}{e}{\sigma} \refinestepto{m} \tmForceWithElab{n}{k}{e}{\sigma}}
      {n \not\ge m
      \qquad e \refinestepto{m} e'}
      \qquad
      \infer
      {\tmForceWithElab{n}{k}{e}{\sigma} \refinestepto{m} \tmForceWithElab{n}{k}{e}{\sigma}}
      {n \ge m
      \qquad e \stepto e'}
      \\[1em]
      \infer
      {\tmForceWithElab{n}{k}{\tmCtxSuspElab{n}{k}{\hat\Psi}{a^n}}{\sigma} \refinestepto{m} \sub{\sigma}{a^n}}
      {n \ge m}
      \\[1em]
      \infer
      {\tmStoreElab{n}{k}{e} \refinestepto{m} \tmStoreElab{n}{k}{e}}
      {n \not\ge m
      \qquad e \refinestepto{m} e'}
      \qquad
      \infer
      {\tmStoreElab{n}{k}{e} \refinestepto{m} \tmStoreElab{n}{k}{e}}
      {n \ge m
      \qquad e \stepto e'}
      \\[1em]
      \infer
      {\tmLoadElab{n}{l}{x}{e}{f} \refinestepto{m} \tmLoadElab{n}{l}{x}{e'}{f}}
      {l \not\ge m
      \qquad e \refinestepto{m} e'}
      \qquad
      \infer
      {\tmLoadElab{n}{l}{x}{e}{f} \refinestepto{m} \tmLoadElab{n}{l}{x}{e'}{f}}
      {l \ge m
      \qquad e \stepto e'}
      \\[1em]
      \infer
      {\tmLoadElab{n}{l}{x}{v}{f} \refinestepto{m} \tmLoadElab{n}{l}{x}{v}{f'}}
      {l \not\ge m
      \qquad f \refinestepto{m} f'}
      \qquad
      \infer
      {\tmLoadElab{n}{l}{x}{a^m}{f} \refinestepto{m} \tmLoadElab{n}{l}{x}{a^m}{f'}}
      {l \ge m
      \qquad f \refinestepto{m} f'}
      \\[1em]
      \infer
      {\tmTLam{\alpha}{K_k}{e} \refinestepto{m} \tmTLam{\alpha}{K_k}{e'}}
      {e \refinestepto{m} e'}
      \qquad
      \infer
      {\tmTApp{e}{A_k} \refinestepto{m} \tmTApp{e'}{A_k}}
      {e \refinestepto{m} e'}
      \\[1em]
      \infer
      {\tmLam{x}{A_k}{e} \refinestepto{m} \tmLam{x}{A_k}{e'}}
      {e \refinestepto{m} e'}
      \qquad
      \infer
      {\tmApp{e}{f} \refinestepto{m} \tmApp{e'}{f}}
      {e \refinestepto{m} e'}
      \qquad
      \infer
      {\tmApp{a^m}{f} \refinestepto{m} \tmApp{a^m}{f'}}
      {f \refinestepto{m} f'}
    \end{array}
  \]
  \caption{Small-step semantics of \elevator{}}
  \label{fig:small-step}
\end{figure}

The only place where term evaluation depends on template reduction is
\[
  \infer
  {\tmCtxSuspElab{k}{l}{\hat\Psi}{e} \stepto \tmCtxSuspElab{k}{l}{\hat\Psi}{e'}}
  {e \refinestepto{k} e'}
\]
Here, we go into the body of the template to determine if any
template reduction is possible. During the process, it is crucial to keep
the ambient mode \(k\), as illustrated at the beginning of this section,
to decide whether to evaluate a subterm.

The other transition, from template reduction to term evaluation, happens in the following rule, for example:
\[
  \infer
  {\tmForceWithElab{n}{k}{e}{\sigma} \refinestepto{m} \tmForceWithElab{n}{k}{e}{\sigma}}
  {n \ge m
    \qquad e \stepto e'}
\]
In this rule, \(e\) is a term that generates a template. If this generator lives
in an accessible mode \(n\) from the ambient mode \(m\) (i.e. \(n \geq m\)),
we evaluate the generator to eventually produce a resulting template.
Then, we apply
\[
  \infer
  {\tmForceWithElab{n}{k}{\tmCtxSuspElab{n}{k}{\hat\Psi}{a^n}}{\sigma} \refinestepto{m} \sub{\sigma}{a^n}}
  {n \ge m}
\]
to splice the resulting template into a surrounding template.

On the other hand, if the inner term lives in a mode \(n\) where \(n \not\geq m\),
then the ambient mode \(m\) cannot access \(n\). Therefore, we treat the inner term
as a syntactic structure (i.e.\ a part of the surrounding template) for the template
reduction instead of a target of the term evaluation. Thus, under this condition,
the template reduction rule for \(\tmForceWithElab{n}{k}{e}{\sigma}\) becomes
a congruence rule:
\[
  \infer
  {\tmForceWithElab{n}{k}{e}{\sigma} \refinestepto{m} \tmForceWithElab{n}{k}{e}{\sigma}}
  {n \not\ge m
    \qquad e \refinestepto{m} e'}
\]

This branching based on mode ordering occurs multiple times in the template reduction
rules, and it enables syntactic comparison of two templates (and terms) based on
accessible memories. The ability to syntactically compare terms according to memory
accessibility is crucial for capturing the previous notion of code
(e.g.\ as in~\cite{Davies:ACM01}) and is essential for proving \Cref{thm:mode-safety},
which ensures that memory accessibility is respected by the operational semantics.

Now we show that this operational semantics is type-safe. First, it preserves typing.
\begin{theorem}[Preservation]\label{thm:preservation}\mbox{}
  \begin{enumerate}
  \item If \(\Gamma \vdash e : A_k\) and \(e \stepto e'\) then \(\Gamma \vdash e' : A_k\)
  \item If \(\Gamma \vdash e : A_k\) and \(e \refinestepto{m} e'\) then \(\Gamma \vdash e' : A_k\)
  \end{enumerate}
\end{theorem}
\begin{proof}
  By induction on small-step derivations, applying~\ref{lem:substitution} (Substitution Lemma) as needed.
\end{proof}
This operational semantics satisfies the progress property as well:
\begin{theorem}[Progress]\label{thm:progress}\mbox{}
  \begin{enumerate}
  \item If \(\Gamma \vdash e : A_k\) then either \(e \stepto e'\) or \(e\) is a weak normal form.
  \item If \(\Gamma \vdash e : A_k\) then either \(e \refinestepto{m} e'\) or \(e\) is a normal template in mode \(m\).
  \end{enumerate}
\end{theorem}
\begin{proof}
  By induction on typing derivation.
\end{proof}

\subsection{Equational Theory Over a Mode and Mode Safety}
The substructurality of assumptions and type safety does not suffice to ensure
that the preorder, i.e.\ memory accessibility, is respected during evaluation.
In other words, we want to prove a theorem stating that during the evaluation of
a term in a mode \(m\), values in any inaccessible memory region \(k\) are irrelevant.
This statement generalizes the non-interference theorem
found in information flow type systems~\cite{Crary:JFP05,Garg:CSFW06},
which asserts that ``private'' data should never affect the execution of the
``public'' part of a program. Our theorem also generalizes the erasure
theorem~\cite{Jang:FSCD24}, which establishes a similar concept but within
a specific class of types.
 
To state this theorem, we require an equivalence relation over
a mode that relates two terms in that mode if and only if
they are syntactically identical except for some differences in
terms living in inaccessible modes from a chosen mode \(n\), which
corresponds to the ``public'' mode for the information flow systems.
As in the information flow systems, we disregard equivalence between
``private'' terms, i.e. terms in inaccessible modes from \(n\).
Equivalence relations over \(n\) between
two well-formed contexts/kinds/types are defined in
\Cref{fig:equiv-contexts,fig:equiv-kinds,fig:equiv-types}.
Equivalence relations over \(n\) between two well-typed terms
and substitutions are defined in
\Cref{fig:equiv-terms,fig:equiv-substitutions}, respectively.

\begin{figure}[ht]
  \[
    \small
    \begin{array}{c}
      \multicolumn{1}{l}{\fbox{\(\vdash \Gamma \simeq_n \Gamma'\)\qquad\parbox{35em}{when well-formed contexts \(\Gamma\) and \(\Gamma'\) have the same syntax except for inaccessible memory regions from \(n\)}}}\\[3em]
      \infer
      {\vdash \ctxEmpty \simeq_n \ctxEmpty}
      {}
      \\[1em]
      \infer
      {\vdash \ctxCons{\ctxMerge{\Gamma_1}{\Gamma_2}}{\alpha{:}K_k} \simeq_n \ctxCons{\ctxMerge{\Gamma'_1}{\Gamma'_2}}{\alpha{:}K'_k}}
      {\Gamma_1 \geq k \geq n
      \qquad \Gamma'_1 \geq k
      \qquad \vdash \Gamma_1 \simeq_n \Gamma'_1
      \qquad \vdash \ctxMerge{\Gamma_1}{\Gamma_2} \simeq_n \ctxMerge{\Gamma'_1}{\Gamma'_2}
      \qquad \Gamma_1 \vdash K_k \simeq_n K'_k \dashv \Gamma'_1}
      \\[1em]
      \infer
      {\vdash \ctxCons{\ctxMerge{\Gamma_1}{\Gamma_2}}{\alpha{:}K_k} \simeq_n \ctxCons{\ctxMerge{\Gamma'_1}{\Gamma'_2}}{\alpha{:}K'_k}}
      {\Gamma_1 \geq k \not\geq n
      \qquad \Gamma'_1 \geq k
      \qquad \vdash \Gamma_1 \simeq_n \Gamma'_1
      \qquad \vdash \ctxMerge{\Gamma_1}{\Gamma_2} \simeq_n \ctxMerge{\Gamma'_1}{\Gamma'_2}
      \qquad \Gamma_1 \vdash K_k \kindCls
      \qquad \Gamma'_1 \vdash K'_k \kindCls}
      \\[1em]
      \infer
      {\vdash \ctxCons{\ctxMerge{\Gamma_1}{\Gamma_2}}{x{:}A_k} \simeq_n \ctxCons{\ctxMerge{\Gamma'_1}{\Gamma'_2}}{x{:}A'_k}}
      {\Gamma_1 \geq k \geq n
      \qquad \Gamma'_1 \geq k
      \qquad \vdash \Gamma_1 \simeq_n \Gamma'_1
      \qquad \vdash \ctxMerge{\Gamma_1}{\Gamma_2} \simeq_n \ctxMerge{\Gamma'_1}{\Gamma'_2}
      \qquad \Gamma_1 \vdash A_k : \kiType{k} \simeq_n A'_k : \kiType{k} \dashv \Gamma'_1}
      \\[1em]
      \infer
      {\vdash \ctxCons{\ctxMerge{\Gamma_1}{\Gamma_2}}{x{:}A_k} \simeq_n \ctxCons{\ctxMerge{\Gamma'_1}{\Gamma'_2}}{x{:}A'_k}}
      {\Gamma_1 \geq k \not\geq n
      \qquad \Gamma'_1 \geq k
      \qquad \vdash \Gamma_1 \simeq_n \Gamma'_1
      \qquad \vdash \ctxMerge{\Gamma_1}{\Gamma_2} \simeq_n \ctxMerge{\Gamma'_1}{\Gamma'_2}
      \qquad \Gamma_1 \vdash A_k : \kiType{k}
      \qquad \Gamma'_1 \vdash A'_k : \kiType{k}}
    \end{array}
  \]
  \caption{Equivalence over mode \(n\) on contexts}
  \label{fig:equiv-contexts}
\end{figure}

\begin{figure}[ht]
  \[
    \small
    \begin{array}{c}
      \multicolumn{1}{l}{\fbox{\(\Gamma \vdash K_k \simeq_n K'_k \dashv \Gamma'\)\qquad\parbox{32em}{when well-formed kinds \(K_k\) (in \(\Gamma\)) and \(K'_k\) (in \(\Gamma'\)) have the same syntax except for inaccessible modes from \(n\)}}}\\[1.5em]
      \infer
      {\Gamma \vdash \kiType{k} \simeq_n \kiType{k} \dashv \Gamma'}
      {k \geq n}
      \\[1em]
      \infer
      {\Gamma \vdash \kiCtxUp{k}{l}{\Psi}{K_l} \simeq_n \kiCtxUp{k}{l}{\Psi'}{K'_l} \dashv \Gamma'}
      {k > \Psi \geq l \geq n
      \qquad \vdash \ctxCons{\Gamma}{\Psi} \simeq_n \ctxCons{\Gamma'}{\Psi'}
      \qquad \ctxCons{\Gamma}{\Psi} \vdash K_l \simeq_n K'_l \dashv \ctxCons{\Gamma'}{\Psi'}}
      \\[1em]
      \infer
      {\kiCtxUp{k}{l}{\Psi}{K_l} \simeq_n \kiCtxUp{k}{l}{\Psi'}{K'_l}}
      {k > \Psi \geq l \not\geq n
      \qquad \vdash \ctxCons{\Gamma}{\Psi} \simeq_n \ctxCons{\Gamma'}{\Psi'}
      \qquad \ctxCons{\Gamma}{\Psi} \vdash K_l \kindCls
      \qquad \ctxCons{\Gamma'}{\Psi'} \vdash K'_l \kindCls}
    \end{array}
  \]
  \caption{Equivalence over mode \(n\) on kinds}
  \label{fig:equiv-kinds}
\end{figure}

\begin{figure}[ht]
  \[
    \small
    \begin{array}{c}
      \multicolumn{1}{l}{\fbox{\(\Gamma \vdash A_k : K_k \simeq_n A'_k : K'_k \dashv \Gamma'\)\qquad\parbox{30em}{when well-formed types \(A_k\) (of \(K_k\) in \(\Gamma\)) and \(A'_k\) (of \(K'_k\) in \(\Gamma'\)) have the same syntax except for inaccessible modes from \(n\)}}}\\[1.5em]
      \infer
      {\Gamma \vdash \tmCtxSusp{\hat\Psi}{A_l} : \kiCtxUp{k}{l}{\Psi}{K_l} \simeq_n \tmCtxSusp{\hat\Psi}{A'_l} : \kiCtxUp{k}{l}{\Psi'}{K'_l}  \dashv \Gamma'}
      {l \geq n
      \qquad \ctxCons{\Gamma}{\Psi} \vdash A_l : K_l \simeq_n A'_l : K'_l \dashv \ctxCons{\Gamma'}{\Psi'}}
      \\[1em]
      \infer
      {\Gamma \vdash \tmCtxSusp{\hat\Psi}{A_l} : \kiCtxUp{k}{l}{\Psi}{K_l} \simeq_n \tmCtxSusp{\hat\Psi}{A'_l} : \kiCtxUp{k}{l}{\Psi'}{K'_l} \dashv \Gamma'}
      {l \not\geq n
      \qquad \ctxCons{\Gamma}{\Psi} \vdash A_l : K_l
      \qquad \ctxCons{\Gamma'}{\Psi'} \vdash A'_l : K'_l}
      \\[1.5em]
      \infer
      {\ctxMerge{\Gamma_1}{\Gamma_2} \vdash \tyForceWith{A_m}{\sigma} : \subNorm{\sigma}{\eraseForSub{\Psi}}{K_k} \simeq_n \tyForceWith{A'_m}{\sigma'} : \subNorm{\sigma'}{\eraseForSub{\Psi'}}{K'_k} \dashv \ctxMerge{\Gamma'_1}{\Gamma'_2}}
      {\begin{array}{ll}
        \Gamma_1 \geq m \quad \Gamma_1' \geq m \quad \vdash \Gamma_1 \simeq_n \Gamma'_1\\
        \Gamma_1 \vdash \kiCtxUp{m}{k}{\Psi}{K_k} \simeq_n \kiCtxUp{m}{k}{\Psi'}{K'_k} \dashv \Gamma'_1
        &
          \vdash \ctxCons{\ctxMerge{\Gamma_1}{\Gamma_2}}{\Psi} \simeq_n \ctxCons{\ctxMerge{\Gamma'_1}{\Gamma'_2}}{\Psi'}
        \\
        \Gamma_1 \vdash A_m : \kiCtxUp{m}{k}{\Psi}{K_k} \simeq_n A'_m : \kiCtxUp{m}{k}{\Psi'}{K'_k} \dashv \Gamma'_1
        &
          \ctxMerge{\Gamma_1}{\Gamma_2} \vdash \sigma : \Psi \simeq_n \sigma' : \Psi' \dashv \ctxMerge{\Gamma'_1}{\Gamma'_2}
      \end{array}}
      \\[2em]
      \infer
      {\Gamma \vdash \tyCtxUp{k}{l}{\Psi}{A_l} : \kiType{k} \simeq_n \tyCtxUp{k}{l}{\Psi'}{A'_l} : \kiType{k} \dashv \Gamma'}
      {k > \Psi \geq l \geq m
      \qquad k > \Psi'
      \qquad \vdash \ctxCons{\Gamma}{\Psi} \simeq_n \ctxCons{\Gamma'}{\Psi'}
      \qquad \ctxCons{\Gamma}{\Psi} \vdash A_l : \kiType{l} \simeq_n A'_l : \kiType{l} \dashv \ctxCons{\Gamma'}{\Psi'}}
      \\[1em]
      \infer
      {\Gamma \vdash \tyCtxUp{k}{l}{\Psi}{A_l} : \kiType{k} \simeq_n \tyCtxUp{k}{l}{\Psi'}{A'_l} : \kiType{k} \dashv \Gamma'}
      {k > \Psi \geq l \not\geq m
      \qquad k > \Psi'
      \qquad \vdash \ctxCons{\Gamma}{\Psi} \simeq_n \ctxCons{\Gamma'}{\Psi'}
      \qquad \ctxCons{\Gamma}{\Psi} \vdash A_l : \kiType{l}
      \qquad \ctxCons{\Gamma'}{\Psi'} \vdash A'_l : \kiType{l}}
      \\[1em]
      \infer
      {\ctxMerge{\Gamma_1}{\Gamma_2} \vdash \tyDown{m}{k}{A_m} : \kiType{k} \simeq_n \tyDown{m}{k}{A'_m} : \kiType{k} \dashv \ctxMerge{\Gamma'_1}{\Gamma'_2}}
      {\Gamma_1 \geq m \geq k
      \qquad \Gamma'_1 \geq m
      \qquad \vdash \Gamma_1 \simeq_n \Gamma'_1
      \qquad \Gamma_1 \vdash A_m : \kiType{m} \simeq_n A'_m : \kiType{m} \dashv \Gamma'_1}
      \\[1em]
      \infer
      {\ctxMerge{\Gamma_1}{\Gamma_2} \vdash \tyForall{\alpha}{K_m}{A_k} : \kiType{k} \simeq_n \tyForall{\alpha}{K'_m}{A'_k} : \kiType{k} \dashv \ctxMerge{\Gamma'_1}{\Gamma'_2}}
      {\begin{array}{ll}
        \Gamma_1 \ge m \ge k \quad \Gamma'_1 \ge m \quad \vdash \Gamma_1 \simeq_n \Gamma'_1
        \\
        \Gamma_1 \vdash K_m \simeq_n K'_m \dashv \Gamma'_1
        & \ctxCons{(\ctxMerge{\Gamma_1}{\Gamma_2})}{\alpha{:}K_m} \vdash A_k : \kiType{k} \simeq_n A'_k : \kiType{k} \dashv \ctxCons{(\ctxMerge{\Gamma'_1}{\Gamma'_2})}{\alpha{:}K'_m}
      \end{array}}
    \end{array}
  \]
  \caption{Selected rules for equivalence over mode \(n\) on types}
  \label{fig:equiv-types}
\end{figure}

\begin{figure}[ht]
  \[
    \small
    \begin{array}{c}
      \multicolumn{1}{l}{\fbox{\(\Gamma \vdash e : A_k \simeq_n e' : A'_k \dashv \Gamma'\)\qquad\parbox{30em}{when well-typed terms \(e\) (of \(A_k\) in \(\Gamma\)) and \(e'\) (of \(A'_k\) in \(\Gamma'\)) have the same syntax except for inaccessible modes from \(n\)}}}\\[1.5em]
      \infer
      {\Gamma \vdash \tmCtxSusp{\hat\Psi}{e} : \tyUp{k}{l}{\Psi}{A_l} \simeq_n \tmCtxSusp{\hat\Psi'}{e'} : \tyUp{k}{l}{\Psi'}{A'_l} \dashv \Gamma'}
      {l \geq n
      \qquad \ctxCons{\Gamma}{\Psi} \vdash e : A_l \simeq_n e' : A'_l \dashv \ctxCons{\Gamma'}{\Psi'}}
      \qquad
      \infer
      {\Gamma \vdash \tmCtxSusp{\hat\Psi}{e} : \tyUp{k}{l}{\Psi}{A_l} \simeq_n \tmCtxSusp{\hat\Psi'}{e'} : \tyUp{k}{l}{\Psi'}{A'_l} \dashv \Gamma'}
      {l \not\geq n
      \qquad \ctxCons{\Gamma}{\Psi} \vdash e : A_l
      \qquad \ctxCons{\Gamma'}{\Psi'} \vdash e' : A'_l}
      \\[0.75em]
      \infer
      {\ctxMerge{\Gamma_1}{\Gamma_2} \vdash \tmForceWith{e}{\sigma} : \subNorm{\sigma}{\eraseForSub{\Psi}}{A_k} \simeq_n \tmForceWith{e'}{\sigma'} : \subNorm{\sigma'}{\eraseForSub{\Psi}}{A'_k} \dashv \ctxMerge{\Gamma'_1}{\Gamma'_2}}
      {\begin{array}{ll}
        \vdash \Gamma_1 \simeq_n \Gamma'_1 \quad \Gamma_1 \ge m \quad \Gamma'_1 \ge m
        & \vdash \ctxCons{\Gamma_2}{\Psi} \simeq_n \ctxCons{\Gamma'_2}{\Psi'}
        \\
        \Gamma_1 \vdash e : \tyUp{m}{k}{\Psi}{A_k} \simeq_n e' : \tyUp{m}{k}{\Psi'}{A'_k} \dashv \Gamma'_1
        \quad \Gamma_1 \vdash e : \tyUp{m}{k}{\Psi}{A_k} \simeq_n e' : \tyUp{m}{k}{\Psi'}{A'_k} \dashv \Gamma'_1
        & \Gamma_2 \vdash \sigma : \Psi \simeq_n \sigma' : \Psi' \dashv \Gamma'_2
      \end{array}}
      \\[1em]
      \infer
      {\ctxMerge{\Gamma}{\Gamma_W} \vdash \tmStore{e} : \tyDown{m}{k}{A_m} \simeq_n \tmStore{e'} : \tyDown{m}{k}{A'_m} \dashv \ctxMerge{\Gamma'}{\Gamma'_W}}
      {\Gamma \geq m
      \qquad \Gamma' \geq m
      \qquad \vdash \Gamma \simeq_n \Gamma'
      \qquad \Gamma \vdash e : A_m \simeq_n e' : A'_m \dashv \Gamma'}
      \\[0.75em]
      \infer
      {\ctxMerge{\Gamma_1}{\Gamma_2} \vdash \tmLoad{x}{e}{f} : B_k \simeq_n \tmLoad{x}{e'}{f'} : B'_k \dashv \ctxMerge{\Gamma'_1}{\Gamma'_2}}
      {
      \begin{array}{ll}
        \vdash \Gamma_1 \simeq_n \Gamma'_1 \quad \Gamma_1 \ge l \ge k \quad \Gamma'_1 \ge l
        \\
        \Gamma_1 \vdash \tyDown{m}{l}{A_m} : \kiType{m} \simeq_n \tyDown{m}{l}{A'_m} : \kiType{m} \dashv \Gamma'_1
        & \vdash \ctxCons{\Gamma_2}{x{:}A_m} \simeq_n \ctxCons{\Gamma'_2}{x{:}A'_m} \ctxCls
        \\
        \Gamma_1 \vdash e : \tyDown{m}{n}{A_m} \simeq_n e' : \tyDown{m}{n}{A'_m} \dashv \Gamma'_1
        & \ctxCons{\Gamma_2}{x{:}A_m} \vdash f : B_k \simeq_n f' : B'_k \dashv \ctxCons{\Gamma'_2}{x{:}A'_m}
      \end{array}
      }
    \end{array}
  \]
  \caption{Selected rules for equivalence over mode \(n\) on terms}
  \label{fig:equiv-terms}
\end{figure}

\begin{figure}[ht]
  \[
    \small
    \begin{array}{c}
      \multicolumn{1}{l}{\fbox{\(\Gamma \vdash \sigma : \Psi \simeq_n \sigma' : \Psi' \dashv \Gamma'\)\qquad\parbox{30em}{when well-typed substitutions \(\sigma\) (of \(\Psi\) in \(\Gamma\)) and \(\sigma'\) (of \(\Psi'\) in \(\Gamma'\)) have the same syntax except for inaccessible modes from \(n\)}}}\\[1.5em]
      \infer
      {\Gamma_W \vdash \subEmpty : \ctxEmpty \simeq_n \subEmpty : \ctxEmpty \dashv \Gamma'_W}
      {}
      \\[1em]
      \infer
      {\ctxMerge{\Gamma_1}{\Gamma_2} \vdash \subCons{\sigma}{\alpha{\mapsto}A_k} : \ctxCons{\Psi}{\alpha{:}K_k} \simeq_n \subCons{\sigma'}{\alpha{\mapsto}A'_k} : \ctxCons{\Psi'}{\alpha{:}K'_k} \dashv \ctxMerge{\Gamma'_1}{\Gamma'_2}}
      {
      \begin{array}{ll}
        \Gamma_1 \geq k \geq n
        \quad \Gamma'_1 \geq k
        \quad \vdash \Gamma_1 \simeq_n \Gamma'_1
        \\
        \Gamma_1 \vdash \subNorm{\sigma}{\eraseForSub{\Psi}}{K_k} \simeq_n \subNorm{\sigma}{\eraseForSub{\Psi'}}{K'_k} \dashv \Gamma'_1
        \\
        \Gamma_1 \vdash A_k : \subNorm{\sigma}{\eraseForSub{\Psi}}{K_k} \simeq_n A'_k : \subNorm{\sigma}{\eraseForSub{\Psi'}}{K'_k} \dashv \Gamma'_1
        & \ctxMerge{\Gamma_1}{\Gamma_2} \vdash \sigma : \Psi \simeq_n \sigma' : \Psi' \dashv \ctxMerge{\Gamma'_1}{\Gamma'_2}
      \end{array}
      }
      \\[0.75em]
      \infer
      {\ctxMerge{\Gamma_1}{\Gamma_2} \vdash \subCons{\sigma}{\alpha{\mapsto}A_k} : \ctxCons{\Psi}{\alpha{:}K_k} \simeq_n \subCons{\sigma'}{\alpha{\mapsto}A'_k} : \ctxCons{\Psi'}{\alpha{:}K'_k} \dashv \ctxMerge{\Gamma'_1}{\Gamma'_2}}
      {
      \begin{array}{lll}
        \Gamma_1 \geq k \not\geq n
        & \Gamma'_1 \geq k
        \quad \vdash \Gamma_1 \simeq_n \Gamma'_1
        \\
        \Gamma_1 \vdash \subNorm{\sigma}{\eraseForSub{\Psi}}{K_k} \kindCls
        & \Gamma'_1 \vdash \subNorm{\sigma}{\eraseForSub{\Psi'}}{K'_k} \kindCls
        \\
        \Gamma_1 \vdash A_k : \subNorm{\sigma}{\eraseForSub{\Psi}}{K_k}
        & \Gamma'_1 \vdash A'_k : \subNorm{\sigma}{\eraseForSub{\Psi'}}{K'_k}
        & \ctxMerge{\Gamma_1}{\Gamma_2} \vdash \sigma : \Psi \simeq_n \sigma' : \Psi' \dashv \ctxMerge{\Gamma'_1}{\Gamma'_2}
      \end{array}
      }
      \\[1em]
      \infer
      {\ctxMerge{\Gamma_1}{\Gamma_2} \vdash \subCons{\sigma}{x{\mapsto_k}e} : \ctxCons{\Psi}{x{:}A_k} \simeq_n \subCons{\sigma'}{x{\mapsto_k}e'} : \ctxCons{\Psi'}{x{:}A'_k} \dashv \ctxMerge{\Gamma'_1}{\Gamma'_2}}
      {
      \begin{array}{ll}
        \Gamma_1 \geq k \geq n
        \quad \Gamma'_1 \geq k
        \quad \vdash \Gamma_1 \simeq_n \Gamma'_1
        \\
        \Gamma_1 \vdash \subNorm{\sigma}{\eraseForSub{\Psi}}{A_k} \simeq_n \subNorm{\sigma}{\eraseForSub{\Psi'}}{A'_k} \dashv \Gamma'_1
        & \vdash \Gamma_2 \simeq_n \Gamma'_2
        \\
        \Gamma_1 \vdash e : \subNorm{\sigma}{\eraseForSub{\Psi}}{A_k} \simeq_n e' : \subNorm{\sigma}{\eraseForSub{\Psi'}}{A'_k} \dashv \Gamma'_1
        & \Gamma_2 \vdash \sigma : \Psi \simeq_n \sigma' : \Psi' \dashv \Gamma'_2
      \end{array}
      }
      \\[0.75em]
      \infer
      {\ctxMerge{\Gamma_1}{\Gamma_2} \vdash \subCons{\sigma}{x{\mapsto_k}e} : \ctxCons{\Psi}{x{:}A_k} \simeq_n \subCons{\sigma'}{x{\mapsto_k}e'} : \ctxCons{\Psi'}{x{:}A'_k} \dashv \ctxMerge{\Gamma'_1}{\Gamma'_2}}
      {
      \begin{array}{lll}
        \Gamma_1 \geq k \not\geq n
        & \Gamma'_1 \geq k
        \quad \vdash \Gamma_1 \simeq_n \Gamma'_1
        \\
        \Gamma_1 \vdash \subNorm{\sigma}{\eraseForSub{\Psi}}{A_k}
        & \Gamma'_1 \vdash \subNorm{\sigma}{\eraseForSub{\Psi'}}{A'_k}
        & \vdash \Gamma_2 \simeq_n \Gamma'_2
        \\
        \Gamma_1 \vdash e : \subNorm{\sigma}{\eraseForSub{\Psi}}{A_k}
        & \Gamma'_1 \vdash e' : \subNorm{\sigma}{\eraseForSub{\Psi'}}{A'_k}
        & \Gamma_2 \vdash \sigma : \Psi \simeq_n \sigma' : \Psi' \dashv \Gamma'_2
      \end{array}
      }
    \end{array}
  \]
  \caption{Equivalence over mode \(n\) on substitutions}
  \label{fig:equiv-substitutions}
\end{figure}

Note that these equivalence relation rules mirror the structure of the well-formedness/typing
rules, but they disregard equivalences between kinds/types/terms that are
well-formed/well-typed at inaccessible modes from \(n\). Thus, we omit
equivalence rules for type and term constructs that do not affect modes.

To illustrate how these equivalences mirror well-formedness/typing depending on
modes, consider the following rule:
\[
  \infer
  {\Gamma \vdash \kiCtxUp{k}{l}{\Psi}{K_l} \simeq_n \kiCtxUp{k}{l}{\Psi'}{K'_l} \dashv \Gamma'}
  {k > \Psi \geq l \geq n
    \qquad \vdash \ctxCons{\Gamma}{\Psi} \simeq_n \ctxCons{\Gamma'}{\Psi'}
    \qquad \ctxCons{\Gamma}{\Psi} \vdash K_l \simeq_n K'_l \dashv \ctxCons{\Gamma'}{\Psi'}}
\]
Here, we repeat the structure of the rule of the well-formedness check for
\(\kiCtxUp{k}{l}{\Psi}{K_l}\) by replacing well-formedness with equivalence,
except that it comes with one extra condition \(\l \geq n\).
This condition tells us \(K_l\) and \(K'_l\) remain in an accessible mode.

Conversely, if \(l \not\geq n\),
\[
  \infer
  {\kiCtxUp{k}{l}{\Psi}{K_l} \simeq_n \kiCtxUp{k}{l}{\Psi'}{K'_l}}
  {k > \Psi \geq l \not\geq n
    \qquad \vdash \ctxCons{\Gamma}{\Psi} \simeq_n \ctxCons{\Gamma'}{\Psi'}
    \qquad \ctxCons{\Gamma}{\Psi} \vdash K_l \kindCls
    \qquad \ctxCons{\Gamma'}{\Psi'} \vdash K'_l \kindCls}
\]
we check only the well-formedness of \(K_l\) and \(K'_l\) instead of
their equivalence as \(l\) is inaccessible from \(n\). In other words,
we treat such well-formed/typed objects in an inaccessible mode are always
equivalent.

These equivalence relations satisfy the following key lemma (\Cref{lem:diamond-modulo-equiv}).
In the statement, \(\stepto^*\) and \(\refinestepto{l}^*\) are the
reflexive-transitive closures of \(\stepto\) and \(\refinestepto{l}\).

\Needspace{7\baselineskip}
\begin{samepage}
\begin{lemma}[Diamond Modulo Equivalence Over a Mode]\label{lem:diamond-modulo-equiv}For \(\Gamma \vdash e : A_k \simeq_n e' : A'_k \dashv \Gamma'\) where \(k \geq n\),
  \begin{itemize}
  \item If \(e \stepto e_1\) and \(e' \stepto e'_1\) and \(e \stepto^* v\) and \(e' \stepto^* v'\),
    then there exist \(e_2\) and \(e'_2\) such that \(e_1 \stepto^* e_2\) and \(e'_1 \stepto^* e'_2\) and \(e_2 \simeq_n e'_2\).
  \item If \(e \refinestepto{l} e_1\) and \(e' \refinestepto{l} e'_1\) and \(e \refinestepto{l}^* a^l\) and \(e' \refinestepto{l}^* a'^l\),
    then there exist \(e_2\) and \(e'_2\) such that \(e_1 \refinestepto{l}^* e_2\) and \(e'_1 \refinestepto{l}^* e'_2\) and \(e_2 \simeq_n e'_2\).
  \end{itemize}
\end{lemma}
\end{samepage}

\begin{proof}
  By induction on the equivalence derivation and the fact that we cannot step further from a weak normal form. This proof also requires a variant of \Cref{lem:substitution} for the equivalence relations.
\end{proof}

This lemma says that if two well-typed terms
(under two equivalent contexts and types over a mode \(n\))
are equivalent over \(n\) then we can reach two equivalent terms over \(n\)
after evaluating both sides. With this lemma, we can immediately prove
the next main theorem.

\begin{theorem}[Mode Safety]\label{thm:mode-safety}
  For \(k \geq n\), if \(\Gamma \vdash e : A_k \simeq_n e' : A'_k \dashv \Gamma'\) and \(e \stepto^* v\) and \(e' \stepto^* v'\),
  \(\Gamma \vdash v : A_k \simeq_n v' : A'_k \dashv \Gamma'\)
\end{theorem}
\begin{proof}
  By applying \Cref{lem:diamond-modulo-equiv} and the fact that we cannot step further from a weak normal form.
\end{proof}

\section{Implementation and Case Studies}\label{sec:implementation}

\subsection{Type Checking in \elevator}
We have implemented a prototype\footnote{\url{https://github.com/Ailrun/adjoint-meta-impl}}
for \elevator. The core type checking algorithm used in the implementation is an 
extention of the algorithmic typing rules for simply typed adjoint
natural deduction~\cite{Jang:FSCD24} to support polymorphism. In particular:

\begin{itemize}
\item The implementation allows non-normal types in type signatures/annotations.
  These are normalized based on the normalization by hereditary substitution.
\item It also includes data type declarations.
\item Depending on mode structure, the type checker allows general recursions or
  primitive recursions on a data type.
\end{itemize}

This type checking algorithm is parametrised by a mode structure.
One can choose any data type \lstinline!m! as a mode structure as
long as it comes with the following functions:
\begin{lstlisting}
  (<=!!) :: m -> m -> Bool
  modeSig :: m -> ElMdSt -> Bool
\end{lstlisting}
Here, \lstinline{(<=!!)} is the preorder on \lstinline!m!,
and \lstinline!modeSig! is the mode signature function to decide
whether a structural rule (whose type is \lstinline!ElMdSt!
(\elevator Mode Structural rules)) is available in a given mode.

\subsection{Interpreter / Runtime Engine in \elevator}
Evaluation in the implementation is based on the abstract machine defined
by~\citet{Jang:FSCD24}, but we adjust the machine with the template reduction
we have in this paper. More specifically, we distinguish two states for
the evaluation of a term and reduction of a template in the machine, and
each state has different rules corresponding to our operational semantics.
Moreover, as we interpret a term to get a final value, our implementation skips
intermediate states and proceeds to the final state as in big-step semantics.

\subsection{Example : Updating an Array In-place}
To highlight the versatility and relevance of \elevator to existing
programming practice, we extend our naive
implementation of linear lists from
\Cref{subsec:motivation-poly-code-gen} to modelling in-place updates
in arrays. 
In this example, we have the mode specification where
we have 3 modes \(\mC\), \(\mP\), and \(\mGF\), which satisfy
\(\mC > \mP > \mGF\). \(\mC\) and \(\mP\) both allow
all substructural rules, but \(\mGF\) allows no substructural rules.
This allows us to define a template depending on the type variable
\(\alpha\) living in \(\mP\), so that we can use \(\alpha\)
to describe the type of the mapping function, which will be used
multiple times.
Comparing to M\oe{}bius~\cite{Jang:POPL22}, the mode \(\mP\) is
like level 1, and \(\mC\) is like level 2. The important difference
is that here our ``level 0'', i.e.\ mode \(\mGF\), does not allow
weakening/contraction.

We use this linearity of \(\mGF\) to encode side-effects
while keeping typical optimizations based on referential transparency
(such as common subexpression elimination) intact. This approach
is based on~\cite{Bernardy:POPL18}, which uses linear types to
handle, for example, array mutation in a referentially transparent setting.

Suppose that we have the following interface:
\begin{lstlisting}
  read : (\alpha: Type$_\mGF$^|$^\mP_\mGF$) -> Nat$_\mGF$ -o f[\alpha]f Array$_\mGF$ -o (f[\alpha]f^|$^\mP_\mGF$v|$^\mP_\mGF$, f[\alpha]f Array$_\mGF$)
  write : (\alpha: Type$_\mGF$^|$^\mP_\mGF$) -> Nat$_\mGF$ -o f[\alpha]f^|$^\mP_\mGF$v|$^\mP_\mGF$ -o f[\alpha]f Array$_\mGF$ -o f[\alpha]f Array$_\mGF$
\end{lstlisting}
Here, \lstinline!f[\alpha]f Array$_\mGF$! represent a tag for an array,
which is used by these interface functions to access the array.
The \lstinline!read \alpha n xs! function reads \lstinline!n!-th element
in the array tagged by \lstinline!xs!, and returns it together with a new
tag for future access to the array. Note that the type of this element is
\lstinline!f[\alpha]f^|$^\mP_\mGF$v|$^\mP_\mGF$!. Indeed, if we read
for example, 0-th element twice, we can use that element twice. Thus to
encode this multi-use correctly, the element type from \lstinline!read!
should come with \lstinline!^|$^\mP_\mGF$v|$^\mP_\mGF$!.
The \lstinline!write \alpha n v xs! function writes \lstinline!v! to
\lstinline!n!-th place in the array tagged by \lstinline!xs!, and returns
a new tag for the array.

With this interface, we can build a pointer to a template that maps each
entry in an array in the in-place manner.
\begin{lstlisting}
mapArrayHelper : Nat$_\mGF$
               -o [I:Nat$_\mGF$^|$^\mP_\mGF$,              % starting position I
                   \alpha:Type$_\mGF$^|$^\mP_\mGF$,        
                   xs:f[\alpha]f Array$_\mGF$,             
                   f:f[\alpha]f^|$^\mP_\mGF$ -> f[\alpha]f^|$^\mP_\mGF$  
                  |- f[\alpha]f Array$_\mGF$]^|$^\mC_\mGF$v|$^\mC_\mGF$
mapArrayHelper n =
  match n with
  | 0 => store t[I, \alpha, xs, f . xs]t
  | n => load YS = mapArrayHelper (n - 1) in
         store t[I, \alpha, xs, f .
           let (x, xs1) = read \alpha f[I]f xs in
           load X = x in
           let xs2 = write \alpha f[I]f (store (f X)) xs1 in
           f[YS]f@(t[1 + f[I]f]t, \alpha, xs2, f)]t

mapArray : Nat$_\mGF$
         -o [\alpha:Type$_\mGF$^|$^\mP_\mGF$, xs:f[\alpha]f Array$_\mGF$, f:f[\alpha]f^|$^\mP_\mGF$ -> f[\alpha]f^|$^\mP_\mGF$ |- f[\alpha]f Array$_\mGF$]^|$^\mC_\mGF$v|$^\mC_\mGF$
amap n =
  load YS = mapArrayHelper n in
  store t[\alpha, xs, f . f[YS]f@(t[0]t, \alpha, xs, f)]t
\end{lstlisting}
In this program, \lstinline!mapArrayHelper n! generates a program that maps \(n\) elements
in \lstinline!xs! starting from its \lstinline!I!-th element using mapping function
\lstinline!f!. \lstinline!mapArray n! generates a program by instantiating the program
generated by \lstinline!mapArrayHelper n! with \lstinline!0! so that it covers first \(n\)
elements. When we call \lstinline!mapArray! with the expected size of an input array, the
generated program will map all elements in the array.

\section{Related Work}
\paragraph{Metaprogramming Systems Based on S4 Modal Logic}
More than two decades ago, \citet{Davies:ACM01} observe that the necessity (box) modality
\(\tyBoxSymb\) in the modal logic S4 can be used to distinguish generated
code fragments from the program generator itself. This key observation provides
a logical foundation for quasi-quotation systems. Since then, this
system has been extended to various type constructs and metaprogramming
concepts, including open code~\cite{Nanevski:TOCL08}, System-F style
polymorphism and pattern matching~\cite{Jang:POPL22}, and dependent
types~\cite{Boespflug:LFMTP11,Pientka:LICS19,Hu:JFP23,Hu:ESOP24}.

Our system borrows the idea of contextualized polymorphic types from
the contextual box types~\cite{Nanevski:TOCL08} and its polymorphic
version~\cite{Jang:POPL22} but provides a more fine-grained control
over memory \--- whether we want to use a memory entry unrestrictedly,
at-most-once, at-least-once, or exactly-once. We also provide an explicit
statement about memory independence \Cref{thm:mode-safety} (Mode Safety),
where \citet{Jang:POPL22} do not state such a guarantee on their
levels (which corresponds to our modes). On the other hand,
pattern matching on code as in their system remains as future work for
our system.

\paragraph{Adjoint Logic}
\citet{Benton:CSL94} first observed that a pair of adjoint functors can be used
as modalities switching between two different logics, linear and intuitionistic.
\citet{Benton:CSL94} also provides a syntactic theory in both sequent calculus
and natural deduction forms. This system has the advantage that it preserves
not only inhabitedness, but also the term structures and the reductions of
its sublogics (linear logic and intuitionistic logic).
\citet{Reed09} extends \citet{Benton:CSL94}'s sequent calculus
to a multi-logic version that connects any number of logics based on
a preorder between logics (modes) and coined the name ``adjoint logic''.
In this generalization, \citet{Reed09} also observes that the necessity
modality \(\tyBoxSymb\) can be decomposed into an adjoint pair and
seamlessly embedded into adjoint logic. This observation opens up the 
possibility to apply adjoint logic to metaprogramming in the spirit
of~\citet{Davies:ACM01}. Our work makes this observation concrete with the design of \elevator.

\citet{Licata:LFCS16,Licata:FSCD17} provide a category-theoretical
description of a adjoint sequent calculus together with a generalization to
general 2-categories of modes (of which a preorder of modes is an instance).
This description includes categorical semantics as well as an equational
theory of proof terms in the adjoint sequent calculus. While this allows us
to describe even BI (bunched implications), modes in their system are too
general to be easily checked in practice. Thus, in applications of adjoint
logic to message-passing style concurrency and session types,
 \citet{Pruiksma:JLAMP21} and \citet{Pruiksma24phd}
go back to preorder-based systems. In \elevator and \citet{Jang:FSCD24}, we also use a pre-order on modes
for practical reasons.

\elevator is an extension of adjoint natural
deduction~\cite{Jang:FSCD24}, which is a preorder extension of
\citet{Benton:CSL94}'s natural deduction calculus. In \citet{Jang:FSCD24},
the authors show that the adjoint natural deduction system can be clearly
understood as a functional programming language with memory management
guarantees. They provide an algorithmic
typing algorithm and an abstract machine semantics for the
system. Moreover, they prove
some properties in terms of memory management such as garbage-freedom or
erasure (that we can ignore values in inaccessible memory regions) for
a specific class of types (that corresponds to positive propositions in logic).
\elevator extends this work by \citet{Jang:FSCD24} to the polymorphic
setting and defines an operational
semantics that is compatible with existing metaprogramming
systems. Moreover, we generalize the up-shift adjoint modality to
characterize open templates. 

\paragraph{Other Typed Metaprogramming Systems}
One of the most notable typed metaprogramming systems is
MetaML~\cite{Taha:TCS00,Taha:POPL03}. The system allows manipulation of open code
and supports type inference. Their approach is based on environment classifiers,
which are names given to typing contexts. \citet{Kiselyov:APLAS16}
extends this environment classifier
approach with polymorphism and side-effects. As the classifiers abstract out details of contexts,
they reason about environments and their extensions only abstractly.
Therefore, in its nature, in their systems, one cannot reason about specific
type/term variables occurring in a code fragment. This seems problematic in terms of
System-F style polymorphism
, as we cannot annotate a specific
free type variable to be available throughout multiple modes. For example,
the type variable \(\beta\) of \lstinline!mapLinMeta! in
\Cref{subsec:motivation-poly-code-gen} requires such an annotation in order to
be available to describe both garbage-free memory and unrestricted memory.

\citet{Kim:POPL06} characterizes open code using extensible records for
the typing environment of code. As their goal is to extend ML with the macro
system of Lisp, they intentionally allow lexically ill-scoped code fragments.
Also, their extensible records are treated symbolically and do not allow
the weakening of an open code fragment.

\citet{Parreaux:SPLASH17,Parreaux:POPL18} describe a system called Squid,
a metaprogramming framework for Scala. Squid supports code generation and
code analysis using a primitive. This allows a programmer to write a
code transformer without extra traversal over code fragments. However,
although they have code generators that themself are polymorphic,
their system cannot generate a polymorphic code fragment, as their contexts
do not have type variables.

\citet{Xie:POPL22} provide a formal description for Typed Template Haskell.
They combine polymorphic metaprogramming with type class constraint resolution,
which is one of the key features of Haskell. However, this system has not
been integrated with Linear Haskell~\cite{Bernardy:POPL18}.

In general, these systems do not support polymorphism and memory management
together.




\paragraph{Other Substructural Systems Based on Multiple Modes}
Recently a number of multi-modal systems~\cite{McBride16,Atkey:LICS18,Moon:ESOP21,Choudhury:POPL21,Wood22esop,Abel:ICFP23}
supporting both S4 necessity modal type and substructural types
have been proposed. These systems largely differ in their expressive power
and their intended applications but share one common property in terms of
S4 necessity modal type; they evaluate under the modality. This evaluation
under the necessity modality breaks the distinction between code and
programs and makes them not fit well for metaprogramming.


\section{Conclusion}
This paper describes \elevator{}, an adjoint foundation for multi-staged
metaprogramming that, for the first time, provides memory management guarantees
for typed polymorphic metaprogramming. Furthermore, this foundation, even
without its memory management guarantees, brings a more fine-grained distinction
between code and program, which allows us to implement more efficient code
generators. This foundation realizes this by generalizing simply typed adjoint
type systems~\cite{Pruiksma:JLAMP21,Pruiksma24phd,Jang:FSCD24} to a
contextualized, polymorphically-typed system and developing a small-step
semantics that is compatible with metaprogramming.

This paper also introduces a new notion of safety, called ``mode safety'',
which states that memory accessibility specified as a preorder on memory regions
is respected during evaluation. We prove that \elevator{} satisfies  mode
safety as well as the type safety and substructurality of variable references.

One possible extension of this system is to add pattern matching on
templates similar to ~\citet{Jang:POPL22}. In fact,
being able to syntactically compare normal templates is essential to be
compatible with~\cite{Jang:POPL22} and, in general, with pattern matching on
templates as it allows programmers to easily predict the matching branch.
Implementing this extension would introduce the first polymorphic system that
allows both memory safety guarantees and analysis over code fragments.

Mode inference and polymorphism are also useful extensions for practical
programming. By mode inference, we can remove extraneous mode annotations
on types for more compact and readable code. By mode polymorphism (such as
in~\cite{Abel:ICFP20}), we can reduce boilerplates in a program further by
sharing code across multiple modes.


\bibliography{ext.bib}


\begin{thebibliography}{37}


\ifx \showCODEN    \undefined \def \showCODEN     #1{\unskip}     \fi
\ifx \showDOI      \undefined \def \showDOI       #1{#1}\fi
\ifx \showISBNx    \undefined \def \showISBNx     #1{\unskip}     \fi
\ifx \showISBNxiii \undefined \def \showISBNxiii  #1{\unskip}     \fi
\ifx \showISSN     \undefined \def \showISSN      #1{\unskip}     \fi
\ifx \showLCCN     \undefined \def \showLCCN      #1{\unskip}     \fi
\ifx \shownote     \undefined \def \shownote      #1{#1}          \fi
\ifx \showarticletitle \undefined \def \showarticletitle #1{#1}   \fi
\ifx \showURL      \undefined \def \showURL       {\relax}        \fi
\providecommand\bibfield[2]{#2}
\providecommand\bibinfo[2]{#2}
\providecommand\natexlab[1]{#1}
\providecommand\showeprint[2][]{arXiv:#2}

\bibitem[Abel and Bernardy(2020)]%
        {Abel:ICFP20}
\bibfield{author}{\bibinfo{person}{Andreas Abel} {and}
  \bibinfo{person}{Jean-Philippe Bernardy}.} \bibinfo{year}{2020}\natexlab{}.
\newblock \showarticletitle{A Unified View of Modalities in Type Systems}.
\newblock \bibinfo{journal}{\emph{Proc. ACM Program. Lang.}}
  \bibinfo{volume}{4}, \bibinfo{number}{ICFP} (\bibinfo{date}{aug}
  \bibinfo{year}{2020}).
\newblock
\urldef\tempurl%
\url{https://doi.org/10.1145/3408972}
\showDOI{\tempurl}


\bibitem[Abel et~al\mbox{.}(2023)]%
        {Abel:ICFP23}
\bibfield{author}{\bibinfo{person}{Andreas Abel}, \bibinfo{person}{Nils~Anders
  Danielsson}, {and} \bibinfo{person}{Oskar Eriksson}.}
  \bibinfo{year}{2023}\natexlab{}.
\newblock \showarticletitle{A Graded Modal Dependent Type Theory with a
  Universe and Erasure, Formalized}.
\newblock \bibinfo{journal}{\emph{Proc. {ACM} Program. Lang.}}
  \bibinfo{volume}{7}, \bibinfo{number}{{ICFP'23}} (\bibinfo{year}{2023}),
  \bibinfo{pages}{920--954}.
\newblock


\bibitem[Anand et~al\mbox{.}(2018)]%
        {Anand:ITP18}
\bibfield{author}{\bibinfo{person}{Abhishek Anand}, \bibinfo{person}{Simon
  Boulier}, \bibinfo{person}{Cyril Cohen}, \bibinfo{person}{Matthieu Sozeau},
  {and} \bibinfo{person}{Nicolas Tabareau}.} \bibinfo{year}{2018}\natexlab{}.
\newblock \showarticletitle{Towards Certified Meta-Programming with Typed
  {Template-Coq}}. In \bibinfo{booktitle}{\emph{9th International Conference
  Interactive Theorem Proving {(ITP'18)}}} \emph{(\bibinfo{series}{Lecture
  Notes in Computer Science (LNCS 10895)})}. \bibinfo{publisher}{Springer},
  \bibinfo{pages}{20--39}.
\newblock
\urldef\tempurl%
\url{https://doi.org/10.1007/978-3-319-94821-8\_2}
\showDOI{\tempurl}


\bibitem[Atkey(2018)]%
        {Atkey:LICS18}
\bibfield{author}{\bibinfo{person}{Robert Atkey}.}
  \bibinfo{year}{2018}\natexlab{}.
\newblock \showarticletitle{Syntax and Semantics of Quantitative Type Theory}.
  In \bibinfo{booktitle}{\emph{Proceedings of the 33rd Annual ACM/IEEE
  Symposium on Logic in Computer Science (LICS'18)}}.
  \bibinfo{publisher}{Association for Computing Machinery},
  \bibinfo{address}{New York, NY, USA}, \bibinfo{pages}{56–65}.
\newblock
\showISBNx{9781450355834}
\urldef\tempurl%
\url{https://doi.org/10.1145/3209108.3209189}
\showDOI{\tempurl}


\bibitem[Benton(1995)]%
        {Benton:CSL94}
\bibfield{author}{\bibinfo{person}{P.~N. Benton}.}
  \bibinfo{year}{1995}\natexlab{}.
\newblock \showarticletitle{A mixed linear and non-linear logic: Proofs, terms
  and models}. In \bibinfo{booktitle}{\emph{Computer Science Logic}}
  \emph{(\bibinfo{series}{CSL 1994})},
  \bibfield{editor}{\bibinfo{person}{Leszek Pacholski} {and}
  \bibinfo{person}{Jerzy Tiuryn}} (Eds.). \bibinfo{publisher}{Springer Berlin
  Heidelberg}, \bibinfo{address}{Berlin, Heidelberg},
  \bibinfo{pages}{121--135}.
\newblock
\showISBNx{978-3-540-49404-1}


\bibitem[Benton et~al\mbox{.}(1998)]%
        {Benton:JFP98}
\bibfield{author}{\bibinfo{person}{P.~N. Benton}, \bibinfo{person}{Gavin~M.
  Bierman}, {and} \bibinfo{person}{Valeria de Paiva}.}
  \bibinfo{year}{1998}\natexlab{}.
\newblock \showarticletitle{Computational Types from a Logical Perspective}.
\newblock \bibinfo{journal}{\emph{J. Funct. Program.}} \bibinfo{volume}{8},
  \bibinfo{number}{2} (\bibinfo{year}{1998}), \bibinfo{pages}{177--193}.
\newblock


\bibitem[Bernardy et~al\mbox{.}(2017)]%
        {Bernardy:POPL18}
\bibfield{author}{\bibinfo{person}{Jean-Philippe Bernardy},
  \bibinfo{person}{Mathieu Boespflug}, \bibinfo{person}{Ryan~R. Newton},
  \bibinfo{person}{Simon Peyton~Jones}, {and} \bibinfo{person}{Arnaud
  Spiwack}.} \bibinfo{year}{2017}\natexlab{}.
\newblock \showarticletitle{Linear Haskell: Practical Linearity in a
  Higher-Order Polymorphic Language}.
\newblock \bibinfo{journal}{\emph{Proc. ACM Program. Lang.}}
  \bibinfo{volume}{2}, \bibinfo{number}{POPL} (\bibinfo{date}{dec}
  \bibinfo{year}{2017}).
\newblock
\urldef\tempurl%
\url{https://doi.org/10.1145/3158093}
\showDOI{\tempurl}


\bibitem[Boespflug and Pientka(2011)]%
        {Boespflug:LFMTP11}
\bibfield{author}{\bibinfo{person}{Mathieu Boespflug} {and}
  \bibinfo{person}{Brigitte Pientka}.} \bibinfo{year}{2011}\natexlab{}.
\newblock \showarticletitle{Multi-level Contextual Modal Type Theory}. In
  \bibinfo{booktitle}{\emph{6th International Workshop on Logical Frameworks
  and Meta-languages: Theory and Practice (LFMTP'11)}}
  \emph{(\bibinfo{series}{Electronic Proceedings in Theoretical Computer
  Science (EPTCS)}, Vol.~\bibinfo{volume}{71})},
  \bibfield{editor}{\bibinfo{person}{Gopalan Nadathur} {and}
  \bibinfo{person}{Herman Geuvers}} (Eds.). \bibinfo{pages}{29--43}.
\newblock
\urldef\tempurl%
\url{https://doi.org/10.4204/EPTCS.71.3}
\showDOI{\tempurl}


\bibitem[Choudhury et~al\mbox{.}(2021)]%
        {Choudhury:POPL21}
\bibfield{author}{\bibinfo{person}{Pritam Choudhury}, \bibinfo{person}{Harley
  Eades~III}, \bibinfo{person}{Richard~A. Eisenberg}, {and}
  \bibinfo{person}{Stephanie Weirich}.} \bibinfo{year}{2021}\natexlab{}.
\newblock \showarticletitle{A Graded Dependent Type System with a Usage-Aware
  Semantics}.
\newblock \bibinfo{journal}{\emph{Proc. ACM Program. Lang.}}
  \bibinfo{volume}{5}, \bibinfo{number}{POPL} (\bibinfo{date}{jan}
  \bibinfo{year}{2021}).
\newblock
\urldef\tempurl%
\url{https://doi.org/10.1145/3434331}
\showDOI{\tempurl}


\bibitem[Crary et~al\mbox{.}(2005)]%
        {Crary:JFP05}
\bibfield{author}{\bibinfo{person}{Karl Crary}, \bibinfo{person}{Aleksey
  Kliger}, {and} \bibinfo{person}{Frank Pfenning}.}
  \bibinfo{year}{2005}\natexlab{}.
\newblock \showarticletitle{A Monadic Analysis of Information Flow Security
  with Mutable State}.
\newblock \bibinfo{journal}{\emph{Journal of Functional Programming}}
  \bibinfo{volume}{15}, \bibinfo{number}{2} (\bibinfo{year}{2005}),
  \bibinfo{pages}{249–291}.
\newblock
\urldef\tempurl%
\url{https://doi.org/10.1017/S0956796804005441}
\showDOI{\tempurl}


\bibitem[Davies and Pfenning(2001)]%
        {Davies:ACM01}
\bibfield{author}{\bibinfo{person}{Rowan Davies} {and} \bibinfo{person}{Frank
  Pfenning}.} \bibinfo{year}{2001}\natexlab{}.
\newblock \showarticletitle{A modal analysis of staged computation}.
\newblock \bibinfo{journal}{\emph{J. ACM}} \bibinfo{volume}{48},
  \bibinfo{number}{3} (\bibinfo{year}{2001}), \bibinfo{pages}{555--604}.
\newblock
\urldef\tempurl%
\url{https://doi.org/10.1145/382780.382785}
\showDOI{\tempurl}


\bibitem[Garg and Pfenning(2006)]%
        {Garg:CSFW06}
\bibfield{author}{\bibinfo{person}{Deepak Garg} {and} \bibinfo{person}{Frank
  Pfenning}.} \bibinfo{year}{2006}\natexlab{}.
\newblock \showarticletitle{Non-interference in Constructive Authorization
  Logic}. In \bibinfo{booktitle}{\emph{19th IEEE Computer Security Foundations
  Workshop (CSFW'06)}}. \bibinfo{pages}{11 pp.--296}.
\newblock
\urldef\tempurl%
\url{https://doi.org/10.1109/CSFW.2006.18}
\showDOI{\tempurl}


\bibitem[Hu et~al\mbox{.}(2023)]%
        {Hu:JFP23}
\bibfield{author}{\bibinfo{person}{Jason Z.~S. Hu}, \bibinfo{person}{Junyoung
  Jang}, {and} \bibinfo{person}{Brigitte Pientka}.}
  \bibinfo{year}{2023}\natexlab{}.
\newblock \showarticletitle{Normalization by evaluation for modal dependent
  type theory}.
\newblock \bibinfo{journal}{\emph{Journal of Functional Programming}}
  \bibinfo{volume}{33} (\bibinfo{year}{2023}), \bibinfo{pages}{e7}.
\newblock
\urldef\tempurl%
\url{https://doi.org/10.1017/S0956796823000060}
\showDOI{\tempurl}


\bibitem[Hu and Pientka(2024)]%
        {Hu:ESOP24}
\bibfield{author}{\bibinfo{person}{Jason Z.~S. Hu} {and}
  \bibinfo{person}{Brigitte Pientka}.} \bibinfo{year}{2024}\natexlab{}.
\newblock \showarticletitle{Layered Modal Type Theory}. In
  \bibinfo{booktitle}{\emph{Programming Languages and Systems}},
  \bibfield{editor}{\bibinfo{person}{Stephanie Weirich}} (Ed.).
  \bibinfo{publisher}{Springer Nature Switzerland}, \bibinfo{address}{Cham},
  \bibinfo{pages}{52--82}.
\newblock
\showISBNx{978-3-031-57262-3}


\bibitem[Jang et~al\mbox{.}(2022)]%
        {Jang:POPL22}
\bibfield{author}{\bibinfo{person}{Junyoung Jang}, \bibinfo{person}{Samuel
  G\'{e}lineau}, \bibinfo{person}{Stefan Monnier}, {and}
  \bibinfo{person}{Brigitte Pientka}.} \bibinfo{year}{2022}\natexlab{}.
\newblock \showarticletitle{M\oe{}bius: Metaprogramming Using Contextual Types:
  The Stage Where System F Can Pattern Match on Itself}.
\newblock \bibinfo{journal}{\emph{Proc. ACM Program. Lang.}}
  \bibinfo{volume}{6}, \bibinfo{number}{POPL} (\bibinfo{date}{jan}
  \bibinfo{year}{2022}).
\newblock
\urldef\tempurl%
\url{https://doi.org/10.1145/3498700}
\showDOI{\tempurl}


\bibitem[Jang et~al\mbox{.}(2024)]%
        {Jang:FSCD24}
\bibfield{author}{\bibinfo{person}{Junyoung Jang}, \bibinfo{person}{Sophia
  Roshal}, \bibinfo{person}{Frank Pfenning}, {and} \bibinfo{person}{Brigitte
  Pientka}.} \bibinfo{year}{2024}\natexlab{}.
\newblock \showarticletitle{{Adjoint Natural Deduction}}. In
  \bibinfo{booktitle}{\emph{9th International Conference on Formal Structures
  for Computation and Deduction (FSCD 2024)}} \emph{(\bibinfo{series}{Leibniz
  International Proceedings in Informatics (LIPIcs)},
  Vol.~\bibinfo{volume}{299})}, \bibfield{editor}{\bibinfo{person}{Jakob
  Rehof}} (Ed.). \bibinfo{publisher}{Schloss Dagstuhl -- Leibniz-Zentrum
  f{\"u}r Informatik}, \bibinfo{address}{Dagstuhl, Germany},
  \bibinfo{pages}{15:1--15:23}.
\newblock
\showISBNx{978-3-95977-323-2}
\showISSN{1868-8969}
\urldef\tempurl%
\url{https://doi.org/10.4230/LIPIcs.FSCD.2024.15}
\showDOI{\tempurl}


\bibitem[Kim et~al\mbox{.}(2006)]%
        {Kim:POPL06}
\bibfield{author}{\bibinfo{person}{Ik-Soon Kim}, \bibinfo{person}{Kwangkeun
  Yi}, {and} \bibinfo{person}{Cristiano Calcagno}.}
  \bibinfo{year}{2006}\natexlab{}.
\newblock \showarticletitle{A polymorphic modal type system for lisp-like
  multi-staged languages}. In \bibinfo{booktitle}{\emph{33rd ACM SIGPLAN-SIGACT
  Symposium on Principles of Programming Languages(POPL '06)}}.
  \bibinfo{publisher}{ACM Press}, \bibinfo{address}{New York, NY, USA},
  \bibinfo{pages}{257--268}.
\newblock
\urldef\tempurl%
\url{https://doi.org/10.1145/1111037.1111060}
\showDOI{\tempurl}


\bibitem[Kiselyov et~al\mbox{.}(2016)]%
        {Kiselyov:APLAS16}
\bibfield{author}{\bibinfo{person}{Oleg Kiselyov}, \bibinfo{person}{Yukiyoshi
  Kameyama}, {and} \bibinfo{person}{Yuto Sudo}.}
  \bibinfo{year}{2016}\natexlab{}.
\newblock \showarticletitle{Refined Environment Classifiers - Type- and
  Scope-Safe Code Generation with Mutable Cells}. In
  \bibinfo{booktitle}{\emph{14th Asian Symposium on Programming Languages and
  Systems ({APLAS}'16)}} \emph{(\bibinfo{series}{Lecture Notes in Computer
  Science}, Vol.~\bibinfo{volume}{10017})},
  \bibfield{editor}{\bibinfo{person}{Atsushi Igarashi}} (Ed.).
  \bibinfo{pages}{271--291}.
\newblock
\urldef\tempurl%
\url{https://doi.org/10.1007/978-3-319-47958-3\_15}
\showDOI{\tempurl}


\bibitem[Licata and Shulman(2016)]%
        {Licata:LFCS16}
\bibfield{author}{\bibinfo{person}{Daniel~R. Licata} {and}
  \bibinfo{person}{Michael Shulman}.} \bibinfo{year}{2016}\natexlab{}.
\newblock \showarticletitle{Adjoint Logic with a 2-Category of Modes}. In
  \bibinfo{booktitle}{\emph{Logical Foundations of Computer Science}},
  \bibfield{editor}{\bibinfo{person}{Sergei Artemov} {and}
  \bibinfo{person}{Anil Nerode}} (Eds.). \bibinfo{publisher}{Springer
  International Publishing}, \bibinfo{address}{Cham},
  \bibinfo{pages}{219--235}.
\newblock
\showISBNx{978-3-319-27683-0}
\urldef\tempurl%
\url{https://doi.org/10.1007/978-3-319-27683-0_16}
\showDOI{\tempurl}


\bibitem[Licata et~al\mbox{.}(2017)]%
        {Licata:FSCD17}
\bibfield{author}{\bibinfo{person}{Daniel~R. Licata}, \bibinfo{person}{Michael
  Shulman}, {and} \bibinfo{person}{Mitchell Riley}.}
  \bibinfo{year}{2017}\natexlab{}.
\newblock \showarticletitle{{A Fibrational Framework for Substructural and
  Modal Logics}}. In \bibinfo{booktitle}{\emph{2nd International Conference on
  Formal Structures for Computation and Deduction (FSCD 2017)}}
  \emph{(\bibinfo{series}{Leibniz International Proceedings in Informatics
  (LIPIcs)}, Vol.~\bibinfo{volume}{84})},
  \bibfield{editor}{\bibinfo{person}{Dale Miller}} (Ed.).
  \bibinfo{publisher}{Schloss Dagstuhl--Leibniz-Zentrum fuer Informatik},
  \bibinfo{address}{Dagstuhl, Germany}, \bibinfo{pages}{25:1--25:22}.
\newblock
\showISBNx{978-3-95977-047-7}
\showISSN{1868-8969}
\urldef\tempurl%
\url{https://doi.org/10.4230/LIPIcs.FSCD.2017.25}
\showDOI{\tempurl}


\bibitem[McBride(2016)]%
        {McBride16}
\bibfield{author}{\bibinfo{person}{Conor McBride}.}
  \bibinfo{year}{2016}\natexlab{}.
\newblock \showarticletitle{I Got Plenty o' Nuttin'}.
\newblock In \bibinfo{booktitle}{\emph{A List of Successes That can Change the
  World{---}Essays Dedicated to {P}hilip {Wadler} on the Occasion of His 60th
  Birthday}}, \bibfield{editor}{\bibinfo{person}{Sam Lindley},
  \bibinfo{person}{Conor McBride}, \bibinfo{person}{Phil Trinder}, {and}
  \bibinfo{person}{Don Sannella}} (Eds.). \bibinfo{publisher}{Springer LNCS
  9600}, \bibinfo{pages}{207--233}.
\newblock


\bibitem[Moon et~al\mbox{.}(2021)]%
        {Moon:ESOP21}
\bibfield{author}{\bibinfo{person}{Benjamin Moon},
  \bibinfo{person}{Harley~Eades III}, {and} \bibinfo{person}{Dominic Orchard}.}
  \bibinfo{year}{2021}\natexlab{}.
\newblock \showarticletitle{Graded Modal Dependent Type Theory}. In
  \bibinfo{booktitle}{\emph{30th European Symposium on Programming
  ({ESOP'21)}}} \emph{(\bibinfo{series}{Lecture Notes in Computer Science},
  Vol.~\bibinfo{volume}{12648})}. \bibinfo{publisher}{Springer},
  \bibinfo{pages}{462--490}.
\newblock


\bibitem[Nanevski et~al\mbox{.}(2008)]%
        {Nanevski:TOCL08}
\bibfield{author}{\bibinfo{person}{Aleksandar Nanevski}, \bibinfo{person}{Frank
  Pfenning}, {and} \bibinfo{person}{Brigitte Pientka}.}
  \bibinfo{year}{2008}\natexlab{}.
\newblock \showarticletitle{Contextual modal type theory}.
\newblock \bibinfo{journal}{\emph{{ACM Transactions on Computational Logic}}}
  \bibinfo{volume}{9}, \bibinfo{number}{3} (\bibinfo{year}{2008}),
  \bibinfo{pages}{1--49}.
\newblock


\bibitem[Parreaux et~al\mbox{.}(2017)]%
        {Parreaux:SPLASH17}
\bibfield{author}{\bibinfo{person}{Lionel Parreaux}, \bibinfo{person}{Amir
  Shaikhha}, {and} \bibinfo{person}{Christoph~E. Koch}.}
  \bibinfo{year}{2017}\natexlab{}.
\newblock \showarticletitle{Squid: type-safe, hygienic, and reusable
  quasiquotes}. In \bibinfo{booktitle}{\emph{SCALA@SPLASH}}.
  \bibinfo{publisher}{{ACM}}, \bibinfo{pages}{56--66}.
\newblock


\bibitem[Parreaux et~al\mbox{.}(2018)]%
        {Parreaux:POPL18}
\bibfield{author}{\bibinfo{person}{Lionel Parreaux}, \bibinfo{person}{Antoine
  Voizard}, \bibinfo{person}{Amir Shaikhha}, {and}
  \bibinfo{person}{Christoph~E. Koch}.} \bibinfo{year}{2018}\natexlab{}.
\newblock \showarticletitle{Unifying analytic and statically-typed
  quasiquotes}.
\newblock \bibinfo{journal}{\emph{{PACMPL}}} \bibinfo{volume}{2},
  \bibinfo{number}{{POPL}} (\bibinfo{year}{2018}),
  \bibinfo{pages}{13:1--13:33}.
\newblock
\urldef\tempurl%
\url{https://doi.org/10.1145/3158101}
\showDOI{\tempurl}


\bibitem[Pfenning and Davies(2001)]%
        {Pfenning01mscs}
\bibfield{author}{\bibinfo{person}{Frank Pfenning} {and} \bibinfo{person}{Rowan
  Davies}.} \bibinfo{year}{2001}\natexlab{}.
\newblock \showarticletitle{A judgmental reconstruction of modal logic}.
\newblock \bibinfo{journal}{\emph{Mathematical Structures in Computer Science}}
  \bibinfo{volume}{11}, \bibinfo{number}{4} (\bibinfo{year}{2001}),
  \bibinfo{pages}{511--540}.
\newblock


\bibitem[Pientka et~al\mbox{.}(2019)]%
        {Pientka:LICS19}
\bibfield{author}{\bibinfo{person}{Brigitte Pientka}, \bibinfo{person}{Andreas
  Abel}, \bibinfo{person}{Francisco Ferreira}, \bibinfo{person}{David
  Thibodeau}, {and} \bibinfo{person}{Rebecca Zucchini}.}
  \bibinfo{year}{2019}\natexlab{}.
\newblock \showarticletitle{A Type Theory for Defining Logics and Proofs}. In
  \bibinfo{booktitle}{\emph{34th {IEEE}/ ACM Symposium on Logic in Computer
  Science (LICS'19)}}. \bibinfo{publisher}{IEEE Computer Society},
  \bibinfo{pages}{1--13}.
\newblock


\bibitem[Pruiksma(2024)]%
        {Pruiksma24phd}
\bibfield{author}{\bibinfo{person}{Klaas Pruiksma}.}
  \bibinfo{year}{2024}\natexlab{}.
\newblock \emph{\bibinfo{title}{Adjoint Logic with Applications}}.
\newblock \bibinfo{thesistype}{Ph.\,D. Dissertation}. \bibinfo{school}{Carnegie
  Mellon University}.
\newblock
\newblock
\shownote{In preparation}.


\bibitem[Pruiksma and Pfenning(2021)]%
        {Pruiksma:JLAMP21}
\bibfield{author}{\bibinfo{person}{Klaas Pruiksma} {and} \bibinfo{person}{Frank
  Pfenning}.} \bibinfo{year}{2021}\natexlab{}.
\newblock \showarticletitle{A Message-passing Interpretation of Adjoint Logic}.
\newblock \bibinfo{journal}{\emph{Journal of Logical and Algebraic Methods in
  Programming}}  \bibinfo{volume}{120} (\bibinfo{date}{4}
  \bibinfo{year}{2021}), \bibinfo{pages}{100637}.
\newblock
\showISSN{2352-2208}
\urldef\tempurl%
\url{https://doi.org/10.1016/J.JLAMP.2020.100637}
\showDOI{\tempurl}


\bibitem[Reed(2009)]%
        {Reed09}
\bibfield{author}{\bibinfo{person}{Jason Reed}.}
  \bibinfo{year}{2009}\natexlab{}.
\newblock \bibinfo{title}{A Judgemental Deconstruction of Modal Logic}.
  (\bibinfo{year}{2009}).
\newblock
\urldef\tempurl%
\url{http://www.cs.cmu.edu/\~jcreed/papers/jdml.pdf}
\showURL{%
\tempurl}


\bibitem[Sheard and Jones(2002)]%
        {Sheard:Haskell02}
\bibfield{author}{\bibinfo{person}{Tim Sheard} {and}
  \bibinfo{person}{Simon~Peyton Jones}.} \bibinfo{year}{2002}\natexlab{}.
\newblock \showarticletitle{Template Meta-programming for {H}askell}. In
  \bibinfo{booktitle}{\emph{ACM SIGPLAN Workshop on Haskell (Haskell'02)}}.
  \bibinfo{publisher}{ACM}, \bibinfo{pages}{1--16}.
\newblock
\showISBNx{1-58113-605-6}
\urldef\tempurl%
\url{https://doi.org/10.1145/581690.581691}
\showDOI{\tempurl}


\bibitem[Taha and Nielsen(2003)]%
        {Taha:POPL03}
\bibfield{author}{\bibinfo{person}{Walid Taha} {and}
  \bibinfo{person}{Michael~Florentin Nielsen}.}
  \bibinfo{year}{2003}\natexlab{}.
\newblock \showarticletitle{Environment classifiers}. In
  \bibinfo{booktitle}{\emph{30th ACM SIGPLAN-SIGACT Symposium on Principles of
  Programming Languages (POPL'03)}}. \bibinfo{publisher}{ACM Press},
  \bibinfo{address}{New Orleans, Louisisana}, \bibinfo{pages}{26--37}.
\newblock
\urldef\tempurl%
\url{https://doi.org/10.1145/640128.604134}
\showDOI{\tempurl}


\bibitem[Taha and Sheard(2000)]%
        {Taha:TCS00}
\bibfield{author}{\bibinfo{person}{Walid Taha} {and} \bibinfo{person}{Tim
  Sheard}.} \bibinfo{year}{2000}\natexlab{}.
\newblock \showarticletitle{{MetaML} and Multi-stage Programming with Explicit
  Annotations}.
\newblock \bibinfo{journal}{\emph{Theoretical Computer Science}}
  \bibinfo{volume}{248}, \bibinfo{number}{1-2} (\bibinfo{date}{Oct.}
  \bibinfo{year}{2000}), \bibinfo{pages}{211--242}.
\newblock
\showISSN{0304-3975}
\urldef\tempurl%
\url{https://doi.org/10.1016/S0304-3975(00)00053-0}
\showDOI{\tempurl}


\bibitem[van~der Walt and Swierstra.(2012)]%
        {vanderWalt:IFL12}
\bibfield{author}{\bibinfo{person}{Paul van~der Walt} {and}
  \bibinfo{person}{Wouter Swierstra.}} \bibinfo{year}{2012}\natexlab{}.
\newblock \showarticletitle{Engineering Proof by Reflection in {A}gda}. In
  \bibinfo{booktitle}{\emph{24th Intern. Symp. on Implementation and
  Application of Functional Languages (IFL)}},
  \bibfield{editor}{\bibinfo{person}{Ralf Hinze}} (Ed.).
  \bibinfo{publisher}{Springer}, \bibinfo{pages}{157--173}.
\newblock


\bibitem[Watkins et~al\mbox{.}(2002)]%
        {Watkins02tr}
\bibfield{author}{\bibinfo{person}{Kevin Watkins}, \bibinfo{person}{Iliano
  Cervesato}, \bibinfo{person}{Frank Pfenning}, {and} \bibinfo{person}{David
  Walker}.} \bibinfo{year}{2002}\natexlab{}.
\newblock \bibinfo{booktitle}{\emph{A Concurrent Logical Framework {I}:
  Judgments and Properties}}.
\newblock \bibinfo{type}{{T}echnical {R}eport} CMU-CS-02-101.
  \bibinfo{institution}{Department of Computer Science, Carnegie Mellon
  University}.
\newblock


\bibitem[Wood and Atkey(2022)]%
        {Wood22esop}
\bibfield{author}{\bibinfo{person}{James Wood} {and} \bibinfo{person}{Robert
  Atkey}.} \bibinfo{year}{2022}\natexlab{}.
\newblock \showarticletitle{A Framework for Substructural Type Systems}. In
  \bibinfo{booktitle}{\emph{31st European Symposium on Programming (ESOP
  2022)}}, \bibfield{editor}{\bibinfo{person}{Ilya Sergey}} (Ed.).
  \bibinfo{publisher}{Springer LNCS 13240}, \bibinfo{address}{Munich, Germany},
  \bibinfo{pages}{376--402}.
\newblock


\bibitem[Xie et~al\mbox{.}(2022)]%
        {Xie:POPL22}
\bibfield{author}{\bibinfo{person}{Ningning Xie}, \bibinfo{person}{Matthew
  Pickering}, \bibinfo{person}{Andres L\"{o}h}, \bibinfo{person}{Nicolas Wu},
  \bibinfo{person}{Jeremy Yallop}, {and} \bibinfo{person}{Meng Wang}.}
  \bibinfo{year}{2022}\natexlab{}.
\newblock \showarticletitle{Staging with {C}lass: a {S}pecification for {T}yped
  {T}emplate Haskell}.
\newblock \bibinfo{journal}{\emph{Proc. ACM Program. Lang.}}
  \bibinfo{volume}{6}, \bibinfo{number}{POPL} (\bibinfo{date}{jan}
  \bibinfo{year}{2022}).
\newblock
\urldef\tempurl%
\url{https://doi.org/10.1145/3498723}
\showDOI{\tempurl}


\end{thebibliography}

\appendix

\end{document}